\documentclass[11pt]{article}
\usepackage[top=1.0in, bottom=1.0in, left=0.8in, right=0.8in]{geometry} 
\usepackage[T1]{fontenc}
\usepackage{textcomp}
\usepackage{dsfont}
\usepackage[frenchmath]{mathastext}  
\usepackage{mathrsfs}
\usepackage{algorithm}
\usepackage{algpseudocode}
\usepackage{amsmath,amssymb}
\usepackage{physics}
\usepackage{bm}
\usepackage{bbm}
\usepackage{graphicx}									
\usepackage[font={small}]{caption}
\usepackage{booktabs}
\usepackage[flushleft]{threeparttable} 
\usepackage{multirow}
\usepackage{subfigure}
\usepackage{tabularx}
\usepackage{float}
\usepackage{upgreek}
\usepackage{enumerate}												
\usepackage{mdwlist}												
\usepackage[dvipsnames]{xcolor}										
\usepackage[style=apa, natbib, uniquelist=false, backend=biber]{biblatex}
\usepackage{hyperref}
	\hypersetup{
            plainpages=false,
		colorlinks   = true,
		citecolor    = RoyalBlue,
		linkcolor    = RubineRed, 
		urlcolor     = RubineRed
	}               
\usepackage{xcolor,mathtools}                      
\mathtoolsset{showonlyrefs=true}       
\usepackage{amsthm}
\usepackage{thmtools}            
	\allowdisplaybreaks                  
	\theoremstyle{plain}
	\newtheorem{theorem}{Theorem}
	   
	\newtheorem{proposition}{Proposition}
	
	\theoremstyle{definition}
	\newtheorem{definition}{Definition}
	
	\newtheorem{remark}{Remark}
	\newtheorem{assumption}{Assumption}

\usepackage{soul}

\usepackage{empheq}

\usepackage{appendix}

\newcommand\Db{\mathds{D}}
\newcommand\Eb{\mathds{E}}
\newcommand\Fb{\mathds{F}}
\newcommand\Gb{\mathds{G}}

\newcommand\Pb{\mathds{P}}

\newcommand\Rb{\mathds{R}}

\newcommand\Nb{\mathds{N}}

\newcommand\Ac{\mathcal{A}}

\newcommand\Dc{\mathcal{D}}

\newcommand\Fc{\mathcal{F}}
\newcommand\Gc{\mathcal{G}}

\newcommand\Lc{\mathcal{L}}
\newcommand\Mc{\mathcal{M}}

\newcommand\xt{\widetilde{x}}
\newcommand\yt{\widetilde{y}}

\renewcommand{\d}{\partial}

\newcommand{\ee}{\mathrm{e}}

\usepackage{todonotes}

\newcommand{\eqlnostar}[2]{\begin{align}\label{#1}#2\end{align}}
\newcommand{\eqstar}[1]{\begin{align*}#1\end{align*}}
\newcommand{\eq}[1]{\ifthenelse{\equal{#1}{*}}
  {\eqstar}
  {\eqlnostar{#1}}
 }
\makeatother

\begin{document}

\title{Efficient calibration of the shifted square-root diffusion model to credit default swap spreads using asymptotic approximations}

\author{Ankush Agarwal \footnote{DSAS, University of Western Ontario, London, Canada. email: \texttt{aagarw93@uwo.ca}}\and Ying Liao \footnote{ASBS, University of Glasgow, Glasgow, United Kingdom. email: \texttt{ying.liao@glasgow.ac.uk}}}
\date{\today}

\maketitle 

\begin{abstract}
\noindent We derive a closed-form approximation for the credit default swap (CDS) spread in the two-dimensional shifted square-root diffusion (SSRD) model using asymptotic coefficient expansion technique to approximate solutions of nonlinear partial differential equations. Specifically, we identify the Cauchy problems associated with two terms in the CDS spread formula that lack analytical solutions and derive asymptotic approximations for these terms. Our approximation does not require the assumption of uncorrelated interest rate and default intensity processes as typically required for calibration in the SSRD model. Through several calibration studies using market data on CDS spread, we demonstrate the accuracy and efficiency of our proposed formula. 
\smallskip

\end{abstract}

\section{Introduction}
In this work, we apply the asymptotic coefficient expansion method of \citet{lorig_analytical_2015} for nonlinear partial differential equations (PDEs) to derive an approximate formula for credit default swap (CDS) spread under the two-dimensional shifted square-root diffusion (SSRD) model of \citet{brigo_credit_2005}, within a stochastic intensity framework. Unlike the original approach, which assumed uncorrelated interest rate and default intensity factors for calibration, we derive our approximation formula and subsequently propose a calibration strategy, without assuming such a condition. Through extensive numerical studies, we show that our approximation formula allows fast calibration to market data and provides accurate estimates of CDS spreads and related survival probabilities.  

The two-dimensional SSRD model is based on the generalised square-root diffusion model introduced by \citet{duffie_modeling_1999}, where both interest rate and default intensity are modelled as Cox-Ingersoll-Ross (CIR) processes. A key feature of the SSRD model is its ability to separate the calibration of interest rate and default intensity processes using bond price and credit spread data. The model has gained significant attention and has been extended in later research. For example, \citet{brigo_comparison_2006} demonstrated its effectiveness in capturing implied volatility patterns in CDS options and compared its performance to a term structure market model. \citet{brigo_credit_2006} further enhanced the model by adding jumps to the stochastic intensity process, enabling better calibration to default swaptions across various strikes and maturities. Later, \citet{brigo_exact_2010} advanced analytical tractability of the extended SSRD model and developed an exact formula for pricing default swaptions.

In contrast to the SSRD model, recent developments in credit derivatives modelling have focused on integrating market dynamics -  such as stock prices, volatilities, and default intensities - while making more restrictive assumptions about interest rate. \citet{carr_jump_2006} introduced the jump-to-default constant elasticity of variance (JDCEV) model, a hybrid credit-equity framework where stock prices follow a diffusion process with the potential to drop to zero at default. In this model, interest rate is assumed to be positive and deterministic, and the default intensity is linked to the instantaneous variance of stock price. \citet{di_francesco_cds_2019} extended this model by introducing a stochastic interest rate, which can be negative, using the Vasicek process, relaxing the assumption of positive and deterministic rate. However, these assumptions, particularly around interest rate, may not fully reflect recent market conditions, such as the Bank of Japan's decision to raise rates in March 2024, signaling an end to the era of negative interest rate.

While the SSRD model provides key advantages, it does not offer a closed-form solution for CDS spreads, requiring approximation methods. \citet{brigo_credit_2005} proposed an analytical approximation using the \textit{Vasicek-mapping} technique to price CDS within the SSRD framework. This method involves mapping the two-dimensional correlated CIR dynamics to a correlated Vasicek model, then calculating the CDS spread using the Vasicek dynamics. In this work, we avoid the mapping approach and derive an asymptotic approximation for the CDS spread directly within the SSRD model. Recently, \citet{lorig_explicit_2017} introduced a unified method for pricing European-style options across various volatility models, building on earlier work of \citet{pagliarani_analytical_2012} and \citet{lorig_analytical_2015}. Their approach uses Taylor's expansions to solve nonlinear PDEs with state-dependent coefficients. We apply their method to derive a second-order asymptotic approximation for CDS spreads, using the CDS pricing framework of \citet{brigo_credit_2005}. Based on this approximation, we propose a new calibration technique that does not assume uncorrelated interest rate and default intensity factors. Through extensive numerical tests with real-world data, we show that our second-order approximation accurately estimates CDS spreads and outperforms the mapping approximation in the SSRD model.

The rest of this paper is organised as follows: Section \ref{sec: model setting} introduces the operational mechanism of CDS and reviews the SSRD model. Section \ref{sec: cds spread} establishes the notations, and Section \ref{sec: asymptotic approximation technique} outlines the asymptotic approximation technique. Section \ref{sec: approximation formula} presents the explicit approximation formulas for the CDS spread and risk-neutral survival probability. Section \ref{sec: calibration approach} details the model calibration process, and Section \ref{sec: calibration results} demonstrates the numerical performance of the proposed method in the calibrated model. The appendices \ref{sec: 2nd expression} to \ref{sec: further calibration - JDCEV} contain the explicit expression for the second-order approximation formula of CDS spread and additional calibration results.

\section{Model and assumptions}
\label{sec: model setting}
A credit default swap (CDS) is a financial derivative that insures against the default risk of a specific entity, called the reference entity. The protection buyer has the right to sell the reference entity's bonds at face value if a default occurs, with the notional principal defining this face value. To maintain this right, the buyer pays premiums to the seller, based on a percentage of the notional principal determined by the CDS spread. The protection seller compensates the buyer upon default, typically paying a percentage of the notional principal known as the loss-given-default, which equals one minus the recovery rate.

For analytical simplicity, we assume frictionless financial markets with no arbitrage opportunities or dividends. Following  \citet{duffie_modeling_1999}, we model an entity's default as an exogenous process that is independent of default-free markets, with no economic or financial indicators predicting its occurrence. This means the default dynamic is entirely governed by the default intensity process. The default intensity at time $t$, symbolised by $\lambda_t$, is defined as the instantaneous probability of the default event. We assume that the default intensity is time-dependent, stochastic, and strictly positive at all times, i.e., $\lambda_t >0$ for $t \ge 0$. 

\begin{assumption}
\label{assump: default}
The default is hypothesised as the first jump in a Cox process (doubly stochastic Poisson process), and the time $\tau$ at which the default occurs, is characterised as the initial jump time within this framework. 
\end{assumption}

Let us introduce a complete filtered probability space $(\Omega, \Gc, \Pb)$, where $\Pb$ denotes the risk-neutral pricing measure based on the short-term interest rate process $r_t$. We define the filtration $\Db := \{\Dc_t,\ t \ge 0 \}$, where $\Dc_t := \sigma (D_u \mid u\le t)$, with a right-continuous process $D_t := \mathds{1}_{\{\tau \le t\}}$, capturing the occurrence and precise timing of a default event. Moreover, we define the filtration $\Fb := \{\Fc_t, t \ge 0 \}$, representing the information of the default-free market up to time $t$. Let $\Gb := \{\Gc_t, t \ge 0 \},$ be an arbitrary filtration on $(\Omega, \Gc, \Pb)$, which is assumed to satisfy $\Gb = \Fb \lor \Db$, i.e., $\Gc_t = \Fc_t \lor \Dc_t$. Therefore, $\Gb$ is the enlarged filtration. This probability space forms the foundation for all subsequent stochastic processes, with expectations calculated with respect to $\Pb$.

\begin{definition}
\label{def: default time}
(\citet{brigo_interest-rate_2001}, Section 22.2.3) Let $\Lambda_t$ denote the cumulative intensity up to time $t$, which is defined through the formula $\Lambda_t := \int_0^t \lambda_s \dd s$. By the property of the Cox process, the default time is permitted to be denoted as $\tau := \Lambda^{-1} (\xi)$, where $\xi$ is a standard exponential random variable independent of $\Fc_t$.
\end{definition}

\begin{remark}
Assuming that the default intensity $\lambda_t$ indicates the instantaneous probability of default in the infinitesimal interval $[t, t + \dd t]$, given no prior default, it can be inferred that
\eqstar{
\lambda_t \dd t = \Pb \bigl(t \le \tau < t + \dd t\ | \tau \ge t \bigr). 
}
Thus, the survival probability at time $t$, under the risk-neutral measure, is given as
\eqlnostar{eq: survival probability}{
\Pb \bigl( \tau \ge t \bigr) = \Eb \Big[\Pb\bigl(\xi \ge \int_0^t \lambda_s \dd s \big| \Fc_t \bigr) \Big] = \Eb \big[\ee^{-\int_0^t \lambda_s \dd s} \big].
}
\end{remark}

Within the two-dimensional shifted square-root diffusion (SSRD) model proposed by \citet{brigo_credit_2005}, the short-term interest rate $r_t$ and the default intensity $\lambda_t$ are modelled as correlated Cox-Ingersoll-Ross (CIR) processes. The model dynamics are given as
\eqlnostar{eq: ssrd}{
\dd r_t &= \alpha_1(\beta_1 - r_t) \dd t + \sigma_1 \sqrt{r_t} \dd W_t^{(1)}, \nonumber \\
\dd \lambda_t &= \alpha_2(\beta_2 - \lambda_t) \dd t + \sigma_2 \sqrt{\lambda_t} \dd W_t^{(2)},
}
where $W_t^{(1)}$ and $W_t^{(1)}$ are standard Brownian motions under $\Pb,$ with instantaneous correlation $\rho \in [-1,1]$. That is, $\dd W_t^{(1)}\dd W_t^{(2)} = \rho \dd t$. To ensure that the interest rate and the default intensity always remain positive, the parameters must satisfy the following Feller conditions
\eqlnostar{eq: term feller}{
2\alpha_1\beta_1 > \sigma_1^2, && 2\alpha_2\beta_2 > \sigma_2^2.
}
where $\alpha_1, \alpha_2, \beta_1, \beta_2, \sigma_1$, $\sigma_2$ are constants. 

\begin{remark}
\label{re: survivial probability - CIR}
In the SSRD model \eqref{eq: ssrd}, the formula for the risk-neutral survival probability is analogous to that for the zero-coupon bond (ZCB) price. Hence, we have
\eqlnostar{eq: survival probability - CIR}{
Q(t, T):= \Eb \big[\ee^{-\int_t^T \lambda_s \dd s} \big] = A(t,T)\ee^{-B(t,T) \lambda_t}, && t \in [0,T],
}
where
\eqstar{
A(t,T) &:= \left( \frac{2h \ee^{\frac{(\alpha_2 + h)(T-t)}{2}}}{2h + (\alpha_2 + h)(\ee^{h(T-t)} - 1)} \right)^{2\alpha_2\beta_2 \over \sigma_2^2},\\
B(t,T) &:= \frac{2 (\ee^{h(T-t)} - 1)}{2h + (\alpha_2 + h)(\ee^{h(T-t)} - 1)},\\
h &:= \sqrt{\alpha_2^2 + 2\sigma_2^2}.
}
\end{remark}

\section{CDS spread}
\label{sec: cds spread}
The CDS spread is determined by equating the present value of all expected future premium payments to that of the expected protection payment involved in the CDS.

\begin{proposition}
\label{pro: cds spread}
Let $t_i$ be the i-th premium payment date, and $M$ be the total number of premium payments such that $0 \le t_0 < t_1 < \cdots < t_M$. Moreover, we denote the maturity as $T$, typically $t_M = T$. Thus, the CDS spread at time $t$ is derived as
\eqlnostar{eq: cds spread}{
R_t = \frac{\mathds{1}_{\{\tau > t \}} (1-\zeta) \int_t^T \Eb \big[\ee^{-\int_t^s (r_u + \lambda_u) \dd u} \lambda_s \big| \Fc_t \big] \dd s}{\mathds{1}_{\{\tau > t \}} \int_t^T \Eb \big[\ee^{-\int_t^s (r_u + \lambda_u) \dd u} \lambda_s \big| \Fc_t \big] (s - t_{N(s)-1}) \dd s 
+ \mathds{1}_{\{\tau > t \}} \sum_{i=N(t)}^M (t_i - t_{i-1}) \Eb \big[ \ee^{-\int_t^{t_i} (r_s + \lambda_s) \dd s} \big| \Fc_t \big]},
}
where $t_{N(t)}$ is the first date among the $t_i$ that follows $t$, and $\zeta$ is the recovery rate.
\end{proposition}

\begin{proof}
On the probability space $(\Omega, \Gc, \Pb)$, assuming a unit notional principal, the discounted value at time $t$ of all expected premium payments based on rate $R_t$ is given as
\eqstar{
\text{Pre} (t, T) 
:= \mathds{1}_{\{\tau > t\}} \Eb \left[\ee^{-\int_t^{\tau} r_s \dd s} (\tau - t_{N(\tau)-1}) R_t \mathds{1}_{\{\tau < T\}} + \sum_{i=N(t)}^M (t_i - t_{i-1}) \ee^{-\int_t^{t_i} r_s \dd s} R_t \mathds{1}_{\{\tau \ge t_i\}} \mid \Gc_t \right],
}
and that of the expected protection payment is given by
\eqstar{
\text{Pro} (t, T)  := \mathds{1}_{\{\tau > t\}} \Eb \left[ (1-\zeta) \ee^{-\int_t^{\tau} r_s \dd s} \mathds{1}_{\{\tau < T\}} \mid \Gc_t \right].
}
By Corollary 5.1.1 and Corollary 5.1.3 in \citet{bielecki_credit_2004}, the above two formulas are transformed into the following form
\eqstar{
\mathrm{Pre} (t, T) 
&= \mathds{1}_{\{\tau > t \}} \left( \Eb \big[\ee^{-\int_t^{\tau} r_s \dd s} (\tau - t_{N(\tau)-1}) R_t \mathds{1}_{\{\tau < T \}} \big| \Gc_t \big] 
+ \Eb \big[ \sum_{i=N(t)}^M (t_i - t_{i-1}) R_t \ee^{-\int_t^{t_i} r_s \dd s} \mathds{1}_{\{\tau > t_i \}} \big| \Gc_t \big] \right)\\
&= \mathds{1}_{\{\tau > t \}} R_t \left (\int_t^T \Eb \big[\ee^{-\int_t^s (r_u + \lambda_u) \dd u} \lambda_s \big| \Fc_t \big] (s - t_{N(s)-1}) \dd s
+ \sum_{i=N(t)}^M (t_i - t_{i-1}) \Eb \big[ \ee^{-\int_t^{t_i} (r_s + \lambda_s) \dd s} \big| \Fc_t \big] \right),\\
\mathrm{Pro} (t, T) 
&= \mathds{1}_{\{\tau > t \}} \Eb \big[ \ee^{-\int_t^{\tau} r_s \dd s} (1-\zeta) \mathds{1}_{\{\tau < T \}} \big| \Gc_t \big]\\
&= \mathds{1}_{\{\tau > t \}}  (1-\zeta) \int_t^T \Eb \big[\ee^{-\int_t^s (r_u + \lambda_u) \dd u} \lambda_s \big| \Fc_t \big] \dd s.
}
CDS spread is the rate $R_t$ which makes the  value of the premium and protection legs in CDS equal to each other. Thus, from the derivations above, we get the formula for CDS spread in \eqref{eq: cds spread}.
\end{proof}

\begin{remark}
Market CDS spread at any time can be obtained by setting $t=0$ in \eqref{eq: cds spread} and adjusting the remaining time to maturities accordingly. Thus, going forward we use the following spread formula
\eqlnostar{eq: cds spread0}{
R\equiv R_0 = \frac{(1-\zeta) \int_0^T \Eb \big[\ee^{-\int_0^s (r_u + \lambda_u) \dd u} \lambda_s \big] \dd s}{ \int_0^T \Eb \big[\ee^{-\int_0^s (r_u + \lambda_u) \dd u} \lambda_s \big] (s - t_{N(s)-1}) \dd s 
+ \sum_{i=1}^M (t_i - t_{i-1}) \Eb \big[ \ee^{-\int_0^{t_i} (r_s + \lambda_s) \dd s} \big]}.
}
\end{remark}

Although the SSRD model provides an effective approach for modelling CDS prices, the risk-neutral expectation formula \eqref{eq: cds spread} lacks a closed-form solution. Thus, we derive the approximation formulas for it using the approximation method and procedure outlined in \citet{pascucci_pde_2011}, \citet{lorig_analytical_2015} and \citet{lorig_explicit_2017}. 

\section{Asymptotic approximation technique}
\label{sec: asymptotic approximation technique}
First, we outline the procedure for approximating the solution to a Cauchy problem, which serves as the foundation for deriving the approximation formula for the CDS spread in \eqref{eq: cds spread}. Consider a Cauchy problem in the following form
\eqlnostar{eq: cauchy}{
(\partial_t + \Ac (t)) u(t,z) &= 0, && t \in [0, T),\ z \in \Rb^d, \\
u(T,z) &= \varphi(z), && z \in \Rb^d.
}
where $\Ac(t)$ is the $d$-dimensional second-order differential operator
\eqlnostar{eq: term A}{
\Ac (t) := \sum_{|\alpha| \le 2} a_{\alpha} (t, z) D_z^\alpha.
}
In the above formula, $\alpha$ is written using the standard multi-index notation and is given as follows
\eqstar{
\alpha := (\alpha_1, \dots, \alpha_d)\in \mathbb{N}_0^d, && |\alpha| := \sum_{i=1}^d \alpha_i.}
Then, the differential operator $D_z^\alpha := \partial_{z_1}^{\alpha_1} \cdots \partial_{z_d}^{\alpha_d}.$ 
We assume that the coefficients $a_\alpha$ are bounded and have globally Lipschitz continuous derivatives up to order $N\in \Nb_0$, with the derivatives having a bounded norm. Moreover, we denote by $(a_{\alpha, n} (t, z))_{0\le n\le N}$ an $N$-th order polynomial expansion for any $t\in [0,T]$, where $a_{\alpha, n} (t, z)$ are polynomials with state-independent $a_{\alpha, 0}$ , i.e, $a_{\alpha, 0} (t, \cdot) = a_{\alpha, 0} (t)$. For any fixed $\bar{z}: \Rb^+ \to \Rb^d$, we define $a_{\alpha, n}$ as the $n$-th order term of the Taylor expansion of $a_\alpha$ in the spatial variables around $\bar{z}(\cdot)$. That is, we set
\eqlnostar{eq: a expansion}{
a_{\alpha, n} (\cdot, z) := \sum_{|\beta|=n}\frac{D_z^\beta a_{\alpha} (\cdot, \bar{z}(\cdot))}{\beta !} (z-\bar{z}(\cdot))^{\beta}, && n\le N, && |\alpha|\le 2,
}
with $\beta ! := \beta_1 ! \cdots \beta_d !$ and $z^\beta := z_1^{\beta_1} \cdots z_d^{\beta_d}$. The above expansion for the coefficients $a_\alpha$ allows the expansion point $\bar{z}$ to evolve in time, which enhances the approximation accuracy for option prices, as noted in \citet{lorig_explicit_2017}.
Let us assume that $\Ac(t)$ can be rewritten as $\Ac(t) = \sum_{n=0}^\infty \Ac_n (t)$, where $\Ac_n (t)$ is given by
\eqlnostar{eq: A expansion}{
\Ac_n (t) := \sum_{|\alpha| \le 2} a_{\alpha, n} (t, z) D_z^\alpha,
}
In light of this expansion, we express $u$ as an infinite sum
\eqlnostar{eq: un}{
u(t, z) = \sum_{n=0}^\infty u_n(t, z).
}
Inserting the above formulas - expanded forms of $\Ac$ and $u$ into \eqref{eq: cauchy} - we find that $(u_n(t, z))_{n\ge 0}$ satisfy the following sequence of Cauchy problems
\begin{align}
(\partial_t + \Ac_0(t)) u_0(t, z) &= 0, && u_0(T, z) = \varphi(z), && z\in \Rb^d, \label{eq: cauchy expansion - 0}\\
(\partial_t + \Ac_0(t)) u_n(t, z) &= -\sum_{h=1}^\infty \Ac_k (t) u_{n-h}(t, z), && u_n(T, z) = 0, && z\in \Rb^d.\label{eq: cauchy expansion - n}
\end{align}
Since functions $a_{\alpha, 0} (t, \cdot)$ depend only on $t$, the operator $\Ac_0(t)$ is elliptic with time-dependent coefficients. It is useful to express $\Ac_0(t)$ in the following form
\eqlnostar{eq: A0}{
\Ac_0 (t) = {1\over 2} \sum_{i, j =1}^d C_{i, j}(t) \partial_{z_i z_j} + \sum_{i=1}^d m_i(t) \partial_{z_i} + \gamma(t),
}
where $C(t) = \bigl(C_{i, j}\bigr)_{1\leq i, j \leq d}(t),$ is a positive definite $d\times d$ matrix, $m(t)=\bigl(m_i\bigr)_{1\leq i\leq d}(t),$ is a $d$-dimensional vector, and $\gamma(t)$ is a scalar function. By \textit{Duhamel's principle}, we can derive $u_0(t, z)$, the solution to \eqref{eq: cauchy expansion - 0}, as follows
\eqlnostar{eq: term u0}{
u_0(t, z) &= \ee ^{\int_t^T \gamma (s) \dd s} \int_{\mathbb R^d} \dd \eta \Gamma_0(t, z, T,\eta) \varphi(\eta),
}
where $\Gamma_0$ is identified as a $d$-dimensional Gaussian density function
\eqlnostar{eq: term Gamma}{
\Gamma_0(t, z, T, \eta) &:= \frac{1}{\sqrt{ (2\pi)^d |\mathbf{C}(t,T)|}} \exp\left(-{1\over 2} \big(\eta-z-\mathbf{m}(t, T)\big)^{\mathrm{T}} \mathbf{C}(t,T)^{-1} \big(\eta-z-\mathbf{m}(t,T)\big)\right),
}
with covariance matrix $\mathbf{C}(t,T)$ and mean vector $z+\mathbf{m}(t,T)$ defined as 
\eqlnostar{eq: term cm}{
\mathbf{C}(t,T) := \int^T_t C(s) \mathrm{d}s, && \mathbf{m}(t,T) := \int^T_t m(s) \mathrm{d}s.
}

\begin{theorem}
\label{thm: un}
For any $n\ge 1$, the function $u_n(t, z)$ satisfying \eqref{eq: cauchy expansion - n}, is given explicitly by
\eqlnostar{eq: term un}{
u_n (t, z) = \Lc_n (t, T) u_0(t, z), && t \in [0, T),\ z \in \Rb^d,
}
where $\Lc_n(t, T)$ is defined as
\eqlnostar{eq: term Ln}{
\Lc_n(t, T) := \sum_{h=1}^n \int_t^T \dd s_1 \int_{s_1}^T \dd s_2 \cdots \int_{s_{h-1}}^T \dd s_h \sum_{i\in I_{n, h}} \Gc_{i_1}(t, s_1) \Gc_{i_2}(t, s_2) \cdots \Gc_{i_h} (t, s_h).
}
In the above, 
\eqstar{
I_{n, h} &:= \{i= (i_1, i_2, \dots, i_h) \in \Nb^h | i_1 + i_2 +\cdots + i_h = n \}, \quad 1 \le h \le n,\\
\Gc_i (t, s) &:=  \sum_{|\alpha|\le 2} a_{\alpha, i} (s, \Mc(t, s)) D_z^\alpha,
}
with $a_{\alpha, n}$ as specified in \eqref{eq: a expansion} and
\eqstar{
\Mc(t,s) := z + \mathbf{m} (t, s) +\mathbf{C} (t, s) \nabla_z.
}
\end{theorem}
We refer the reader to Theorem 3.2 of \citet{lorig_analytical_2015} for a proof of the above result.

\begin{remark}
\label{re: error}
Under some general assumptions (see \citet{lorig_analytical_2015}) on operator $\Ac (t)$, the following bound for the approximation error holds
\eqlnostar{eq: error}{
|u(t, z) - \bar{u}_N(t, z)| \le C (T-t)^{{N\over 2}+1}, && t\in [0,T),\ z\in \Rb^d,
}
where $C$ is a positive constant that only depends on $N$, and $\bar{u}_N(t, z)$ is the $N$-th order approximation, given as
\eqstar{
\bar{u}_N(t, z) := \sum_{n=0}^N u_n(t, z).
}
\end{remark}

\section{Approximation formula}
\label{sec: approximation formula}
Going back to the CDS spread approximation problem discussed in Section \ref{sec: cds spread}, we notice that the following terms in \eqref{eq: cds spread} lack a closed-form solution under the SSRD model
\eqstar{
\Eb \big[ \ee^{-\int_t^{T} (r_s + \lambda_s) \dd s} \big| \Fc_t \big], && \Eb \big[\ee^{-\int_t^T (r_s + \lambda_s) \dd s} \lambda_T \big| \Fc_t \big].
}
This is due to the fact that in the SSRD model, the short-term interest rate $r_t$ and default intensity $\lambda_t$ are assumed to follow correlated CIR processes. In order to get around this hurdle of unavailable closed-form solutions for the above two terms, \citet{brigo_credit_2005} approximated the distribution of correlated CIR processes using correlated Vasicek processes, a technique we call the \textit{Vasicek-mapping} technique. However, the above two terms also satisfy parabolic Cauchy problems of the type \eqref{eq: cauchy}. Thus, we can employ the asymptotic approximation technique in \citet{lorig_analytical_2015} to derive their higher-order approximations, which can potentially be more accurate than their corresponding approximations obtained from the \textit{Vasicek-mapping} technique.

To facilitate the analysis, we introduce functions $v(t, x, y, T)$ and $h(t, x, y, T)$, which are defined as follows
\begin{align}
v(t, x, y, T) &:= \Eb \big[\ee^{-\int_t^T (\ee^{-\alpha_1 s}x_s+\mathrm{e}^{-\alpha_2 s}y_s)\dd s} \big| x_t=x, y_t=y \big], && t \in [0,T],
\label{eq: term v}\\
h(t, x, y, T) &:= \Eb \big[\ee^{-\int_t^T(\ee^{-\alpha_1 s} x_s+\ee^{-\alpha_2 s} y_s) \dd s } y_T \big|x_t=x, y_t=y\big], && t \in [0,T],
\label{eq: term h}
\end{align}
with the transformations 
\eqlnostar{eq: term trans}{
x_t := \ee^{\alpha_1 t} r_t, && y_t := \ee^{\alpha_2 t} \lambda_t.
}
Therefore, the SSRD model \eqref{eq: ssrd} is transformed to
\eqlnostar{eq: ssrd t}{
\dd x_t &= \alpha_1\beta_1\ee^{\alpha_1 t}\dd t+\sigma_1\sqrt{\ee^{\alpha_1 t}x_t}\dd W_t^{(1)}, \nonumber \\
\dd y_t &= \alpha_2\beta_2\ee^{\alpha_2 t}\dd t+\sigma_2\sqrt{\ee^{\alpha_2 t}y_t}\dd W_t^{(2)},\\
\dd W^{(1)}_t \dd W^{(2)}_t &= \rho \dd t \nonumber.
}
According to the \textit{Feynman-Kac principle}, \eqref{eq: term v} and \eqref{eq: term h} are the solutions to the following Cauchy problems, respectively,
\eqlnostar{eq: eqvh}{
(\partial_t + \Ac (t)) v(t,x,y,T)=0, && v(T,x,y,T) = 1,\\
(\partial_t + \Ac (t)) h(t,x,y,T)=0, && h(T,x,y,T) = y,
}
where the operator $\Ac(t)$ is given by
\eqlnostar{eq: term AA}{
\Ac (t) =a(t,x,y)\partial_{xx}+b(t,x,y)\partial_{yy}+c(t,x,y)\partial_{xy}+\kappa(t,x,y)\partial_{x}+k(t,x,y)\partial_{y}+\gamma(t,x,y),
}
with the coefficients $a$, $b$, $c$, $\kappa$, $k$, and $\gamma$ specified accordingly as,
\eqlnostar{eq: coefficient function}{
&a(t,x,y) := {1\over 2} \sigma_1^2 \ee^{\alpha_1 t}x, && b(t,x,y) := {1\over 2} \sigma_2^2 \ee^{\alpha_2 t}y, &&c(t,x,y):= \rho\sigma_1\sigma_2 \ee^{\frac{(\alpha_1+\alpha_2)t}{2}} x^{1\over 2} y^{1\over 2},\nonumber \\
&\kappa(t,x,y):= \alpha_1\beta_2 \ee^{\alpha_1 t}, && k(t,x,y):= \alpha_2\beta_2 \ee^{\alpha_2 t}, &&\gamma(t,x,y):= -(\ee^{-\alpha_1 t}x+\ee^{-\alpha_2 t}y).
}

\begin{remark}
\label{re: v0h0}
Given the transformation \eqref{eq: term trans}, functions $v$ and $h$ satisfy the following equations
\eqstar{
v(0, x, y, T) = \Eb \big[ \ee^{-\int_0^{T} (r_s + \lambda_s) \dd s} \big], && h(0, x, y, T) = \ee^{\alpha_2 T} \Eb \big[\ee^{-\int_0^T (r_s + \lambda_s) \dd s} \lambda_T \big].
}
\end{remark}

\begin{proposition}
\label{pro: cds n}
Under the framework established in Section \ref{sec: model setting}, \ref{sec: asymptotic approximation technique} and \ref{sec: approximation formula}, the N-th order approximation formula for the CDS spread at time $0$ derived in \eqref{eq: cds spread0}, is given as follows
\eqlnostar{eq: cds approximation}{
R_N = \frac{(1-\zeta) \int_0^T \ee^{-\alpha_2 s} \sum_{n=0}^N h_n(0, x, y, s) \dd s}{ \int_0^T \ee^{-\alpha_2 s} \sum_{n=0}^N h_n(0, x, y, s) (s - t_{N(s)-1}) \dd s 
+ \sum_{i=1}^M (t_i - t_{i-1}) \sum_{n=0}^N v_n(0, x, y, t_i)},
}
where
\begin{align}
v_0(t,x,y,s) &= \ee^{-x \psi(-\alpha_1,t,s) - \alpha_1\beta_1 \Theta(-\alpha_1, \alpha_1, t, s)} \ee^{-y \psi(-\alpha_2,t,s) - \alpha_2\beta_2 \Theta(-\alpha_2, \alpha_2, t, s)}, && 0 \le t \le s \le T,
\label{eq: term v0}\\
h_0(t,x,y,s) &= v_0(t,x,y,s)\big(y + \alpha_2 \beta_2 \psi(\alpha_2, t,s) \big), && 0 \le t \le s \le T,
\label{eq: term h0}
\end{align}
with
\eqstar{\psi(\alpha,t_1,t_2) &:=\int_{t_1}^{t_2} \ee^{\alpha s} \dd s,\quad 0\le t_1< t_2\le T,\\
\Theta(\alpha,\beta, t_1, t_2) &:= \int_{t_1}^{t_2} \ee^{\alpha s} \psi(\beta, t, s) \dd s,\quad t\le t_1< t_2\le T.
}
\end{proposition}

\begin{proof}
From Remark \ref{re: error} and \ref{re: v0h0}, we know that the N-th order approximation formulas for $\Eb \big[ \ee^{-\int_0^{T} (r_s + \lambda_s) \dd s} \big]$ and $\Eb \big[\ee^{-\int_0^T (r_s + \lambda_s) \dd s} \lambda_T \big]$ are given by
\eqlnostar{eq: full}{
\Eb \big[ \ee^{-\int_0^{T} (r_s + \lambda_s) \dd s} \big] &= v(0, x, y, T) = \sum_{n=0}^N v_n(0, x, y, T) + \mathcal{O} (T^{{N\over 2}+1}),\\
\Eb \big[\ee^{-\int_0^T (r_s + \lambda_s) \dd s} \lambda_T \big] &= \ee^{-\alpha_2 T} h(0, x, y, T) = \ee^{-\alpha_2 T} \sum_{n=0}^N h_n(0, x, y, T) + \mathcal{O} (T^{{N\over 2}+1}).
}
Analogous to \eqref{eq: cauchy expansion - 0} and \eqref{eq: cauchy expansion - n}, $(v_n (t, x, y, T))_{n\ge 0}$ and $(h_n(t, x, y, T))_{n\ge 0}$ are the solutions of the following Cauchy problems
\eqlnostar{eq: eqve}{
(\partial_t + \Ac_0(t)) v_0(t, x, y, T) &= 0, && v_0(T, x, y, T) = 1, \nonumber\\
(\partial_t + \Ac_0(t)) v_n(t, x, t, T) &= -\sum_{h=1}^\infty \Ac_h (t) v_{n-h}(t, x, y, T), && v_n(T, x, y, T) = 0,
}
and
\eqlnostar{eq: eqhe}{
(\partial_t + \Ac_0(t)) h_0(t, x, y, T) &= 0, && h_0(T, x, y, T) = y, \nonumber\\
(\partial_t + \Ac_0(t)) h_n(t, x, t, T) &= -\sum_{h=1}^\infty \Ac_h (t) h_{n-h}(t, x, y, T), && h_n(T, x, y, T) = 0,
}
where $\Ac_0(t)$ is expressed as
\eqstar{
\Ac_0 (t) =a(t,\bar{x},\bar{y})\partial_{xx}+b(t,\bar{x},\bar{y})\partial_{yy}+c(t,\bar{x},\bar{y})\partial_{xy}+\kappa(t,\bar{x},\bar{y})\partial_{x}+k(t,\bar{x},\bar{y})\partial_{y}+\gamma(t,\bar{x},\bar{y}).
}
We define the time-dependent variables $\bar{x}$ and $\bar{y}$ as follows
\eqlnostar{eq: term bar}{
\bar{x}(t, s) := \bar{x}_{\text{fixed}} + \alpha_1\beta_1 \psi(\alpha_1, t, s), && \bar{y}(t, s) := \bar{y}_{\text{fixed}} + \alpha_2\beta_2\psi(\alpha_2, t, s).
}
with $\bar{x}_{\text{fixed}}$ and $\bar{y}_{\text{fixed}}$ are constants, and function $\psi$ is defined as
\eqstar{
\psi(\alpha,t_1,t_2) :=\int_{t_1}^{t_2} \ee^{\alpha s} \dd s,\quad 0\le t_1< t_2\le T.
}
Therefore, the zeroth-order approximation of $v(t, x, y, T)$, with $(\bar{x}_{\text{fixed}}, \bar{y}_{\text{fixed}}) = (x, y)$, is derived as follows
\eqstar{
v_0(t,x,y,T) &= \ee^{\int_t^T \gamma_{0,0}(s)\dd s} \int_{\mathbb R^2} \Gamma_0(t, x, y, T,\eta) \dd \eta,
}
where
\eqstar{
\gamma_{0,0} (s) := \gamma(s, \bar{x}(t, s), \bar{y}(t, s)) =  -\ee^{-\alpha_1 t} \bar{x}(t, s)-\ee^{-\alpha_2 t} \bar{y}(t, s).
}
Since $\Gamma_0$ now is a two-dimensional Gaussian density function, the zeroth-order approximation formula for $v(t, x, y, T)$ is given by
\eqlnostar{eq: v0}{
v_0(t,x,y,T) = \ee^{\int_t^T \gamma_{0,0}(s) \dd s} = \ee^{-x \psi(-\alpha_1,t,T) - \alpha_1\beta_1 \Theta(-\alpha_1, \alpha_1, t, T)} \ee^{-y \psi(-\alpha_2,t,T) - \alpha_2\beta_2 \Theta(-\alpha_2, \alpha_2, t, T)}.
}
Similarly, the zeroth-order approximation formula for $h(t, x, y, T)$ is deduced as follows
\eqlnostar{eq: h0}{
h_0(t,x,y,T) = \ee^{\int_t^T \gamma(s)ds} \int_{\mathbb R^2} \dd \eta \Gamma_0(t, x, y, T,\xi,\omega) y  = v_0(t,x,y,T)\big(y + \alpha_2 \beta_2 \tau(\alpha_2, t,T) \big).
}
This completes the proof. The explicit expressions for the higher-order approximation formulas for both $v_n(t, x, y, T)$ and $h_n(t, x, y, T)$, along with the detailed derivation process, are presented in Appendix \ref{sec: 2nd expression}.
\end{proof}

Analogously, we can derive a closed-form approximation formula for the risk-neutral survival probability, as defined in \eqref{eq: survival probability}, under the transformed SSRD model \eqref{eq: ssrd t}, and use this formula to estimate survival probabilities inferred from the market CDS spreads.

\begin{proposition}
\label{pro: survival probability - PDE}

Let us define the function $\widetilde{Q}(t, y, T):= \Eb \big[\ee^{-\int_t^T \ee^{-\alpha_2 s} y_s \dd s}\big| y_t=y \big]$, which can be seen as the survival probability \eqref{eq: survival probability} after the transformation in \eqref{eq: term trans}. The explicit expression for the approximation formula of the function $\widetilde{Q}$, with $\bar{y}_{\text{fixed}} = y$, is given as
\eqlnostar{eq: survival probability expansion}{
\widetilde{Q}(t, y, T) = \widetilde{Q}_0(t, y, T) + \widetilde{Q}_1(t, y, T), && t \in [0,T],
}
where
\begin{align}
\widetilde{Q}_0(t, y, T) &= \ee ^{-y \psi(-\alpha_2,t,T) - \alpha_2\beta_2 \Theta(-\alpha_2, \alpha_2, t, T)}, \label{eq: Q0}\\
\widetilde{Q}_1(t, y, T) &= \widetilde{Q}_0(t, y, T) \sigma_2^2 \psi(-\alpha_2, t, T) \int_t^T \ee^{-\alpha_2 s} \left(y \psi(\alpha_2, t, s) + \alpha_2\beta_2\Theta(\alpha_2, \alpha_2, t, s)\right) \dd s\\
&+ \widetilde{Q}_0(t, y, T) \sigma_2^4 \psi(-\alpha_2, t, T)^3 \int_t^T \ee^{\alpha_2 s} \left(y \psi(\alpha_2, t, s) + \alpha_2\beta_2\Theta(\alpha_2, \alpha_2, t, s)\right) \dd s. \label{eq: Q1}
\end{align}
\end{proposition}

\begin{proof}
It is known that the function $\widetilde{Q}$ satisfy the following Cauchy problem
\eqstar{
(\partial_t + \widetilde{\Ac} (t)) \widetilde{Q}(t, y, T) = 0, && \widetilde{Q}(T, y, T) = 1,
}
with the operator $\widetilde{\Ac} (t)$ defined as
\eqstar{
\widetilde{\Ac} (t) := {1\over 2} \sigma_2^2 \ee^{\alpha_2 t} y \d_{yy} + \alpha_2\beta_2 \ee^{\alpha_2 t} \d_y - \ee^{-\alpha_2 t} y.
}
Theorem \ref{thm: un} and Remark \ref{re: error} yield the $N$-th order approximation of the following form
\eqstar{
\widetilde{Q}(t, y, T) = \widetilde{Q}_0(t, y, T) + \sum_{n = 1}^N \Lc_n (t, T) \widetilde{Q}_0(t, y, T) + \mathcal{O} (T^{{N\over 2}+1}).
}
Given that the partial derivatives of coefficients in $\widetilde{\Ac} (t)$ are zero when $n>1$, $\widetilde{Q}$ is confined to the first-order approximation, i.e., $Q(t, y, T) = Q_0(t, y, T) + Q_1(t, y, T)$. Applying \eqref{eq: term u0} and \eqref{eq: term Ln}, we obtain the explicit expression for these two terms as presented in \eqref{eq: Q0} and \eqref{eq: Q1}.
\end{proof}

\section{Calibration approach}
\label{sec: calibration approach}

For estimating the market CDS spread using our approximation formula \eqref{eq: cds approximation}, we first need to calibrate the parameters in the transformed SSRD model \eqref{eq: ssrd t} to the market data. Our proposed calibration procedure involves the following three steps:

\begin{itemize}
\item \textit{Step 1: Calibration of the interest rate model.}\newline
We calibrate the parameters of the interest rate process $r_t$ in the SSRD model \eqref{eq: ssrd} separately to the ZCB prices generated from the interest rate swap curve, such as the London interbank offered rate (LIBOR) swap curve. Since $r_t$ follows a CIR process, the ZCB price formula is similar to the one stated in Remark \ref{re: survivial probability - CIR}, which is given as
\eqlnostar{eq: ZCB prices - CIR}{
P(t, T) := \Eb \big[\ee^{-\int_t^T r_s \dd s} \big] = \widehat{A}(t,T)\ee^{-\widehat{B}(t,T) r_t}, && t \in [0,T],
}
where
\eqstar{
\widehat{A} (t,T) &:= \left( \frac{2\widehat{h} \ee^{(\alpha_1 + \widehat{h})(T-t)\over 2}}{2\widehat{h} + (\alpha_1 + \widehat{h})(\ee^{\widehat{h} (T-t)} - 1)} \right)^{2\alpha_1\beta_1 \over \sigma_1^2},\\
\widehat{B} (t,T) &:= \frac{2 (\ee^{\widehat{h} (T-t)} - 1)}{2\widehat{h} + (\alpha_1 + \widehat{h})(\ee^{\widehat{h} (T-t)} - 1)},\\
\widehat{h} &:= \sqrt{\alpha_1^2 + 2\sigma_1^2}.
}
This calibration step can be formulated as an optimisation problem where the aim is to minimise the difference between the model given  prices and the market data. We employ the unweighted nonlinear least squares method, for which the objective function is formalised as follows
\eqlnostar{eq: Non-LSE - ZCB}{
\min_{(\alpha_1, \beta_1, \sigma_1)}  F(\alpha_1, \beta_1, \sigma_1) = \sum_{i=1}^{N_\text{mkt}} |P (0,T_i) - \widehat{P}_i|^2, \, \text{s.t. }  \alpha_1, \sigma_1 >0,\ 2\alpha_1\beta_1 > \sigma_1^2.
}
Here, $N_\text{mkt}$ is the number of ZCB prices from the market and $\widehat{P}_i$ is the market ZCB price observed with maturity $T_i$. We use the \textit{Nelder-Mead} method for solving the above optimisation problem. After completing this step, we obtain the calibrated values for parameters  $\alpha_1$, $\beta_1$ and $\sigma_1.$  $r_0$ is observed directly from the market.

\item \textit{Step 2: Matching the model ZCB price with its approximation.} \newline
Since we use the transformed SSRD model \eqref{eq: ssrd t} to generate the approximation formula for the CDS spread, we need to ensure that the ZCB prices derived from the models for $r_t$ and $x_t$ are consistent. This is an important calibration step, as it ensures the accuracy of our approximation formula for CDS spreads. Recall that in the transformed SSRD model, interest rate $r_t$ is transformed into $\ee^{-\alpha_1 t} x_t$. According to Proposition \ref{pro: survival probability - PDE}, we can express the ZCB price in terms of $x_t$, with $\bar{x}_{\text{fixed}} = x$, as follows
\eqlnostar{eq: ZCB prices - PDE}{
\widetilde{P}(t, x, T):= \Eb \big[\ee^{-\int_t^T \ee^{-\alpha_1 s} x_s \dd s} \big| x_t=x \big] = \widetilde{P}_0(t, y, T) + \widetilde{P}_1(t, y, T), && t \in [0,T],
}
where
\eqstar{
\widetilde{P}_0(t, x, T) &= \ee ^{-x \psi(-\alpha_1,t,T) - \alpha_1\beta_1 \Theta(-\alpha_1, \alpha_1, t, T)}, \\
\widetilde{P}_1(t, x, T) &= \widetilde{P}_0(t, x, T) \sigma_1^2 \psi(-\alpha_1, t, T) \int_t^T \ee^{-\alpha_1 s} \left(x \psi(\alpha_1, t, s) + \alpha_1\beta_1\Theta(\alpha_1, \alpha_1, t, s)\right) \dd s\\
&+ \widetilde{P}_0(t, x, T) \sigma_1^4 \psi(-\alpha_1, t, T)^3 \int_t^T \ee^{\alpha_1 s} \left(x \psi(\alpha_1, t, s) + \alpha_1\beta_1\Theta(\alpha_1, \alpha_1, t, s)\right) \dd s.
}
To maintain consistency in the ZCB prices across the two representations, we introduce a volatility parameter $\widehat{\sigma}_1$ such that
\eqlnostar{eq: particular volatility}{
P(0, T)^{(\alpha_1, \beta_1, \sigma_1)} = \widetilde{P}(0, x_0, T)^{(\alpha_1, \beta_1, \widehat{\sigma}_1)}.
}
In the above, $P(0, T)^{(\alpha_1, \beta_1, \sigma_1)} $ is the ZCB price generated from the original SSRD model \eqref{eq: ssrd}, while $\widetilde{P}(0, x_0, T)^{(\alpha_1, \beta_1, \widehat{\sigma}_1)}$ is the ZCB price generated the transformed SSRD model \eqref{eq: ssrd} with volatility parameter $\sigma_1$ being replaced with $\widehat{\sigma}_1$. This step is similar to the ``particular Vasicek volatility'' calculation in \citet[Page 16]{brigo_credit_2005}. Given that volatility is a positive constant, in the above equation we can replace $\widehat{\sigma}^2_1$ with $\widetilde{\sigma}_1$ instead and solve a quadratic equation in $\widetilde{\sigma}_1$. If the quadratic equation does have solutions, we select the positive square root of the smaller positive $\widetilde{\sigma}_1$ as the volatility parameter $\widehat{\sigma}_1$. However, if the quadratic equation does not have a real solution, we choose the value of $\widehat{\sigma}_1$ that minimises the difference between $P(0, T)^{(\alpha_1, \beta_1, \sigma_1)}$ and $\widetilde{P}(0, x_0, T)^{(\alpha_1, \beta_1, \widehat{\sigma}_1)}$ as the particular volatility. 
From a practical viewpoint, the volatility parameter $\widehat{\sigma}_1$ can be interpreted as the time-averaged value of $\sigma_1\sqrt{r_t}$ over the interval $[0, T]$. We only compute and utilise the volatility parameter $\widehat{\sigma}_1$ corresponding to the longest maturity in the CDS spread dataset. For instance, if we use a CDS spread dataset covering maturities from 1 year to 10 years for the final calibration, we calculate and apply $\widehat{\sigma}_1$ for the 10-year term using \eqref{eq: particular volatility}. In the following calibration steps, we employ $\widehat{\sigma}_1$ as the volatility parameter for the interest rate model, instead of $\sigma_1$ obtained in the first step.

\item \textit{Step 3: Final calibration.} \newline
In this step, we calibrate the parameters  $\Xi := (\alpha_2, \beta_2, \sigma_2, \lambda_0, \rho)$ in the transformed SSRD model \eqref{eq: ssrd t} to the market data on CDS spread using the first-order approximation formula in \eqref{eq: cds approximation}. Using the first-order approximation formula ensures that this step is computationally quick. Moreover, we observe that no loss of precision occurred even when calculating the final estimate using a second-order approximation formula. In this step, we employ a weighted nonlinear least squares method to minimise the difference between the market CDS spreads and the model CDS spreads, where the weights are chosen to account for different levels of reliability or variability in the market data. The objective function is given as
\eqlnostar{eq: Non-LSE}{
\min_{\Xi}  F(\Xi) = \sum_{i=1}^{N_\text{mkt}} \omega_i |R_i - \widehat{R}_i|^2, \, \text{s.t.} \alpha_2, \sigma_2, \lambda_0 >0,\  -1\leq \rho \leq 1,
}
where $N_\text{mkt}$ denotes the number of market CDS spread values used, $\omega_i$ denotes the weight of each term, $R_i$ denotes the model CDS spread with maturity $T_i$, and $\widehat{R}_i$ denotes the market CDS spread observed with maturity $T_i$. In the case where we assume $\rho = 0$, the objective function simplifies to 
\eqstar{
\min_{\widehat{\Xi}} F(\widehat{\Xi}) = \sum_{i=1}^{N_\text{mkt}} \omega_i |R_i - \widehat{R}_i|^2, \, \text{s.t. } \alpha_2, \sigma_2, \lambda_0 >0.
}
When presenting the numerical results in Section \ref{sec: calibration results}, we specify the choice of weights depending on the available data.
\end{itemize}

After obtaining all the parameter estimates from the calibration procedure, we compute the estimate for the CDS spread using \eqref{eq: cds spread0} with $N=2$. We also compute the estimate of the risk-neutral survival probability $\widetilde{Q}$ by using \eqref{eq: survival probability expansion} and compare it with the survival probability $\widehat{Q}$ inferred from the market CDS spreads. Let $\widehat{Q}_i$ denote the survival probability up to time $T_i.$ Then, we have the following bootstrapping formula for the market survival probability
\eqlnostar{eq: market survival probability}{
\sum_{i=1}^j \widehat{R}_i (T_i - T_{i-1}) \widehat{Q}_i = \sum_{i=1}^j (1-\zeta) (\widehat{Q}_i-\widehat{Q}_{i-1}), && 1\le j \le N_{\text{mkt}},
}
where $\widehat{R}_i$ denotes the market CDS spread observed with maturity $T_i$. 

\section{Results}
\label{sec: calibration results}
For testing the accuracy of our approximation formula derived in Section \ref{sec: approximation formula}, we use CDS spread data with a maximum maturity of 10 years from Bloomberg, reported on 8 April 2024, for four major US and European banks. We also conduct a comparison study with the data reported in \citet{di_francesco_cds_2019}, where the main focus is on negative interest rate. In all cases presented here, as well as in Appendix \ref{sec: further calibration - Bloomberg} and Appendix \ref{sec: further calibration - JDCEV}, it is evident that our approximation formula provides an excellent fit to the market data. Furthermore, the survival probabilities estimated using our approximation are highly accurate when compared with those inferred from the market CDS spreads. Based on the reported computational time in Table \ref{tab: time2} and Table \ref{tab: time1}, we can also conclude that our approximation formula facilitates an efficient and fast calibration to the market data. All experiments are conducted on an Apple M2 Chip (8-core CPU, 8-core GPU, 16-core Neural Engine) and 8 GB Unified Memory (RAM).

\subsection{Bloomberg data}
We only report calibration results for JP Morgan Chase \& Co and HSBC Bank PLC here. The results for other banks - Citigroup Inc and Deutsche Bank AG - are reported in Appendix \ref{sec: further calibration - Bloomberg}. Since the CDS contracts for these entities are traded in different currencies, we calibrated the interest rate CIR model to different daily yields curves for ZCBs, selecting secured overnight financing rate (SOFR) swap curve for the US dollar and Euro short-term rate (ESTR) swap curve for the Euro to generate the corresponding zero rates and ZCB prices. The calibration results for the ZCB price are presented in Table \ref{tab: usd} and Table \ref{tab: euro}. It is evident that our calibration approach yields excellent results. 

\begin{table}[htbp]
\centering
\footnotesize
\captionsetup{font=footnotesize, skip = 5pt}
\begin{minipage}{.48\textwidth}
\centering
\begin{threeparttable}
\caption{Calibration results for ZCB prices (SOFR)}
\label{tab: usd}
\begin{tabularx}{.9\textwidth}
{c>{\centering\arraybackslash}X>{\centering\arraybackslash}X>{\centering\arraybackslash}X}
\toprule
Term & Market & Model & Rel. \\
(year) &  &  & Error\\
\midrule
1  & 0.95075  & 0.95263  & 0.1975\% \\ 
2  & 0.91163  & 0.91316  & 0.1683\% \\ 
3  & 0.87691  & 0.87759  & 0.0773\% \\ 
4  & 0.84472  & 0.84430  & 0.0502\% \\ 
5  & 0.81363  & 0.81262  & 0.1241\% \\ 
6  & 0.78327  & 0.78227  & 0.1275\% \\ 
7  & 0.75383  & 0.75311  & 0.0954\% \\ 
8  & 0.72512  & 0.72506  & 0.0086\% \\ 
9  & 0.69751  & 0.69806  & 0.0786\% \\ 
10 & 0.67080  & 0.67207  & 0.1887\% \\ 
\bottomrule
\end{tabularx}
\begin{tablenotes}
\item $\alpha_1 = 0.88422$, $\beta_1 = 0.03816$, $\sigma_1 = 0.09597$, $r_0 = 0.05384$.
\end{tablenotes}
\end{threeparttable}
\end{minipage}%
\hfill
\begin{minipage}{.48\textwidth}
\centering
\footnotesize
\begin{threeparttable}
\caption{Calibration results for ZCB prices (ESTR)}
\label{tab: euro}
\begin{tabularx}{.9\textwidth}
{c>{\centering\arraybackslash}X>{\centering\arraybackslash}X>{\centering\arraybackslash}X}
\toprule
Term & Market & Model & Rel. \\
(year) &  &  & Error\\
\midrule
1  & 0.96620  & 0.96857  & 0.2449\% \\ 
2  & 0.94187  & 0.94392  & 0.2174\% \\ 
3  & 0.92015  & 0.92103  & 0.0962\% \\ 
4  & 0.89965  & 0.89892  & 0.0810\% \\ 
5  & 0.87926  & 0.87739  & 0.2129\% \\ 
6  & 0.85861  & 0.85638  & 0.2601\% \\ 
7  & 0.83778  & 0.83587  & 0.2279\% \\ 
8  & 0.81658  & 0.81586  & 0.0880\% \\ 
9  & 0.79543  & 0.79632  & 0.1123\% \\ 
10 & 0.77382  & 0.77726  & 0.4439\% \\ 
\bottomrule
\end{tabularx}
\begin{tablenotes}
\item $\alpha_1 = 1.59549$, $\beta_1 = 0.02440$, $\sigma_1 = 0.18694$, $r_0 = 0.03963$.
\end{tablenotes}
\end{threeparttable}
\end{minipage}
\end{table}

As reported on Bloomberg, the maturity dates of CDS contracts are fixed on 20 June and 20 December each year. For a CDS contract traded on 8 April 2024, the 6-month maturity corresponds to 20 December 2024, and so on. Thus, we compute the actual maturity (in years) using the Actual/360 day-count convention when computing the volatility parameter $\widehat{\sigma}_1$ (\textit{Step 2} in Section \ref{sec: calibration approach}) and calibrating to the market data on CDS spread. The volatility parameter $\widehat{\sigma}_1$ are $2.214 \times 10^{-4}$ and $3.253 \times 10^{-7}$ for the US dollar and Euro, respectively, calculated based on the longest available actual maturity of CDS spreads.
For this dataset, we choose the following weight formula
\eqstar{
\omega_i := \frac{\frac{1}{|\text{Bid}_i - \text{Ask}_i|}}{\sum_i^{N_\text{mkt}} \frac{1}{|\text{Bid}_i - \text{Ask}_i|}},
}
as it provides the fastest and best calibration results across various different weight choices. Here, $\text{Bid}_i$ and $\text{Ask}_i$ correspond to the bid and ask values of the CDS spread with $i$-th maturity, respectively. We design this weight formula to assign greater importance to terms with higher liquidity, and are more likely to cause significant calibration error. Additionally, we compare the accuracy of our approximation for the CDS spread with the \textit{Vasicek-mapping} technique from \citet{brigo_credit_2005}. In a nutshell, their procedure is composed of three stages:
\begin{enumerate}
\item Calibrate the parameters in the interest rate and default intensity processes separately to the interest rates and credit markets, \textbf{assuming zero correlation between the two processes}.
\item Compute the ``particular Vasicek volatilities'', denoted as \textit{mapping volatilities}, by employing the \textit{mapping equation} from \citet[page 16]{brigo_credit_2005} for both interest rate and default intensity processes. 
\item Use the obtained parameters in the two-dimensional Vasicek model with customised correlation values to compute the estimates of CDS spread, referred to as the \textit{mapping approximations}. 
\end{enumerate}
When computing the \textit{mapping approximations}, the parameter values for the interest rate process remain the same as those presented in Table \ref{tab: usd} and Table \ref{tab: euro}. The \textit{mapping volatilities} of the interest rate model are 0.01915 and 0.02962 for the US dollar and Euro, respectively. The parameter values for the default intensity model are obtained by calibrating the survival probability formula to the market data. These values are reported in Table \ref{tab: vasicek}. The \textit{mapping volatilities} of the default intensity process are 0.01915 for JP Morgan Chase \& Co and 0.0056 for HSBC Bank PLC. The longest available CDS spread actual maturity is used when computing the \textit{mapping volatilities}. For the correlated case, we consider only extreme values of correlation in the \textit{mapping approximations} and report the better of the two results.

\begin{table}[htbp]
\centering
\footnotesize
\begin{threeparttable}
\captionsetup{font=footnotesize, skip = 5pt}
\caption{Parameters obtained from the calibration of default intensity process for computing \textit{mapping approximation}}
\label{tab: vasicek}
\begin{tabular}{lcccc}
\toprule
& $\alpha_2$  & $\beta_2$ & $\sigma_2$ &   $\lambda_0$ \\
\midrule
JP Morgan Chase \& Co & 0.05815 & 0.04013 & 0.06641 &  0.00145\\
HSBC Bank PLC & 0.10298 & 0.02465 & 0.06978 & 0.00090\\ 
\bottomrule
\end{tabular}
\end{threeparttable}
\end{table}

In Table \ref{tab: jp1} and Table \ref{tab: hsbc1}, we present the CDS spread estimates from the SSRD model using asymptotic approximation technique, referred as to the \textit{PDE approximation}, for two different reference entities, along with the \textit{mapping approximations}. In Table \ref{tab: jp2} and Table \ref{tab: hsbc2}, we also present both approximations while assuming no correlation between the interest rate and default intensity processes. As evident from the results in Table \ref{tab: jp1} - Table \ref{tab: hsbc2}, the \textit{mapping approximations} are highly inaccurate for short maturities, whereas the \textit{PDE approximations} remain accurate throughout the CDS spread curve. In general, the relative bid-ask values of CDS spread are higher for short-term maturities compared to maturities greater than 5 years, which leads to the poor performance of the \textit{Vasicek-mapping} technique. Overall, it can be seen that the estimates via our approximation technique are more accurate than \textit{mapping approximations}.

Furthermore, we can observe that the CDS spread approximations are significantly affected by changes in the correlation parameter, by comparing the results in Table \ref{tab: jp1} and Table \ref{tab: hsbc1} with the results in Table \ref{tab: jp2} and Table \ref{tab: hsbc2}. It is apparent that assuming no correlation between the interest rate and default intensity processes leads to poorer calibration performance in the SSRD model. Therefore, our approximation technique, which does not assume zero correlation between the two processes, not only ensures that the model aligns well with the market, but also outperforms the \textit{Vasicek-mapping} technique of \citet{brigo_credit_2005}. We also report the computational time of our calibration procedure in Table \ref{tab: time2}. Compared with the \textit{Vasicek-mapping} technique, which required an average of 30 seconds for calibration, our approach achieves significant accuracy without requiring too much extra time. In addition to JP Morgan Chase \& Co and HSBC Bank PLC, we also provide calibration results for other entities in Table \ref{tab: citi21} - Table \ref{tab: deut22}. For those entities as well, the estimates of CDS spread obtained from our approximation technique remain accurate when compared to the market data.

We also present the estimation results for the risk-neutral survival probability $\widetilde{Q}$ using \eqref{eq: survival probability expansion} in Table \ref{tab: check-jp1} - Table \ref{tab: check-hsbc2}. It is clear that our approximation formula yields very accurate estimates. However, we observe that the estimated relative errors slightly increase with longer maturities, a phenomenon attributed to the error bound of our approximation, which increases with time, as stated in Remark \ref{re: error}.

\begin{table}[H]
\centering
\footnotesize
\captionsetup{font=footnotesize, skip = 5pt}
\caption{Calibration results for JP Morgan Chase \& Co CDS spreads (correlated case)}
\label{tab: jp1}
\begin{threeparttable}
\begin{tabularx}{.8\textwidth}
{c >{\centering\arraybackslash}X >{\centering\arraybackslash}X >{\centering\arraybackslash}X p{.03cm} >{\centering\arraybackslash}X >{\centering\arraybackslash}X}
\toprule
& & \multicolumn{2}{c}{Mapping approximation} \tnote{a} & &\multicolumn{2}{c}{PDE approximation} \tnote{b}\\
\cmidrule{3-4} \cmidrule{6-7}
Term (year) & Market (bps) & Model (bps) & Rel. Error & & Model (bps)& Rel. Error \\
\midrule
0.7 & 16.669 & 13.467 & 19.2093\% &  & 16.369 & 1.7997\% \\
1.2 & 19.742 & 16.728 & 15.2669\% &  & 19.096 & 3.2722\% \\
1.7 & 21.768 & 19.909 & 8.5401\% &  & 21.798 & 0.1378\% \\
2.2 & 23.782 & 23.008 & 3.2546\% &  & 24.477 & 2.9224\% \\
2.7 & 26.369 & 26.025 & 1.3046\% &  & 27.133 & 2.8973\% \\
3.2 & 28.943 & 28.958 & 0.0518\% &  & 29.767 & 2.8470\% \\
3.8 & 31.720 & 32.371 & 2.0523\% &  & 32.900 & 3.7201\% \\
4.3 & 34.497 & 35.122 & 1.8118\% &  & 35.482 & 2.8553\% \\
4.8 & 37.522 & 37.792 & 0.7196\% &  & 38.040 & 1.3805\% \\
5.3 & 40.530 & 40.382 & 0.3652\% &  & 40.573 & 0.1061\% \\
5.8 & 43.331 & 42.894 & 1.0085\% &  & 43.080 & 0.5793\% \\
6.3 & 46.116 & 45.330 & 1.7044\% &  & 45.561 & 1.2035\% \\
6.8 & 48.916 & 47.691 & 2.5043\% &  & 48.016 & 1.8399\% \\
7.3 & 51.701 & 49.979 & 3.3307\% &  & 50.443 & 2.4332\% \\
7.8 & 53.763 & 52.197 & 2.9128\% &  & 52.843 & 1.7112\% \\
8.3 & 55.825 & 54.345 & 2.6511\% &  & 55.214 & 1.0945\% \\
8.8 & 57.887 & 56.426 & 2.5239\% &  & 57.556 & 0.5718\% \\
9.3 & 59.938 & 58.442 & 2.4959\% &  & 59.870 & 0.1135\% \\
9.8 & 62.001 & 60.394 & 2.5919\% &  & 62.153 & 0.2452\% \\
10.3 & 64.051 & 62.285 & 2.7572\% &  & 64.407 & 0.5558\% \\
\bottomrule
\end{tabularx}
\begin{tablenotes}
\item[a] The correlation coefficient used is $\rho = -1$.
\item[b] $\alpha_2 = 0.00126$, $\beta_2 = 1.46292$, $\sigma_2 = 0.00039$, $\lambda_0 = 0.00207$, $\rho = -0.96$.
\end{tablenotes}
\end{threeparttable}

\bigskip

\centering
\footnotesize
\caption{Calibration results for HSBC Bank PLC CDS spreads (correlated case)}
\label{tab: hsbc1}
\begin{threeparttable}
\begin{tabularx}{.8\textwidth}
{c >{\centering\arraybackslash}X >{\centering\arraybackslash}X >{\centering\arraybackslash}X p{.03cm} >{\centering\arraybackslash}X >{\centering\arraybackslash}X}
\toprule
& & \multicolumn{2}{c}{Mapping approximation} \tnote{a} & &\multicolumn{2}{c}{PDE approximation} \tnote{b}\\
\cmidrule{3-4} \cmidrule{6-7}
Term (year) & Market (bps) & Model (bps) & Rel. Error & & Model (bps)& Rel. Error \\
\midrule
0.7 & 14.610 & 10.536 & 27.8850\% & & 14.208 & 2.7515\% \\
1.2 & 16.770 & 14.008 & 16.4699\% & & 16.753 & 0.1014\% \\
1.7 & 19.293 & 17.355 & 10.0451\% & & 19.281 & 0.0622\% \\
2.2 & 21.802 & 20.573 & 5.6371\% & & 21.793 & 0.0413\% \\
2.7 & 23.863 & 23.664 & 0.8339\% & & 24.288 & 1.7810\% \\
3.2 & 25.912 & 26.632 & 2.7786\% & & 26.766 & 3.2958\% \\
3.8 & 29.049 & 30.036 & 3.3977\% & & 29.717 & 2.2996\% \\
4.3 & 32.185 & 32.746 & 1.7430\% & & 32.155 & 0.0932\% \\
4.8 & 34.607 & 35.345 & 2.1325\% & & 34.574 & 0.0954\% \\
5.3 & 37.015 & 37.839 & 2.2261\% & & 36.973 & 0.1135\% \\
5.8 & 39.809 & 40.231 & 1.0601\% & & 39.352 & 1.1480\% \\
6.3 & 42.588 & 42.526 & 0.1456\% & & 41.711 & 2.0593\% \\
6.8 & 45.382 & 44.728 & 1.4411\% & & 44.049 & 2.9373\% \\
7.3 & 48.161 & 46.842 & 2.7387\% & & 46.365 & 3.7292\% \\
7.8 & 49.824 & 48.870 & 1.9147\% & & 48.660 & 2.3362\% \\
8.3 & 51.488 & 50.817 & 1.3032\% & & 50.932 & 1.0799\% \\
8.8 & 53.152 & 52.686 & 0.8767\% & & 53.182 & 0.0564\% \\
9.3 & 54.806 & 54.481 & 0.5930\% & & 55.409 & 1.1002\% \\
9.8 & 56.470 & 56.205 & 0.4693\% & & 57.613 & 2.0241\% \\
10.3 & 58.124 & 57.860 & 0.4542\% & & 59.793 & 2.8714\% \\
\bottomrule
\end{tabularx}
\begin{tablenotes}
\item[a] The correlation coefficient used is $\rho = -1$.
\item[b] $\alpha_2 = 0.00433$, $\beta_2 = 0.39790$, $\sigma_2 = 0.00006$, $\lambda_0 = 0.00176$, $\rho = -0.45395$.
\end{tablenotes}
\end{threeparttable}

\end{table}
\begin{table}[H]
\centering
\footnotesize
\captionsetup{font=footnotesize, skip = 5pt}
\begin{threeparttable}
\caption{Calibration results for JP Morgan Chase \& Co CDS spreads (uncorrelated case)}
\label{tab: jp2}
\begin{tabularx}{.8\textwidth}
{c >{\centering\arraybackslash}X >{\centering\arraybackslash}X >{\centering\arraybackslash}X p{.05cm} >{\centering\arraybackslash}X >{\centering\arraybackslash}X}
\toprule
& & \multicolumn{2}{c}{Mapping approximation} & &\multicolumn{2}{c}{PDE approximation} \tnote{*}\\
\cmidrule{3-4} \cmidrule{6-7}
Term (year) & Market (bps) & Model (bps) & Rel. Error & & Model (bps)& Rel. Error \\
\midrule
0.7 & 16.669 & 13.423 & 19.4733\% & & 16.361 & 1.8477\% \\
1.2 & 19.742 & 16.614 & 15.8444\% & & 19.091 & 3.2975\% \\
1.7 & 21.768 & 19.706 & 9.4726\% & & 21.796 & 0.1286\% \\
2.2 & 23.782 & 22.702 & 4.5412\% & & 24.478 & 2.9266\% \\
2.7 & 26.369 & 25.607 & 2.8898\% & & 27.136 & 2.9087\% \\
3.2 & 28.943 & 28.424 & 1.7932\% & & 29.772 & 2.8643\% \\
3.8 & 31.720 & 31.694 & 0.0820\% & & 32.905 & 3.7358\% \\
4.3 & 34.497 & 34.325 & 0.4986\% & & 35.488 & 2.8727\% \\
4.8 & 37.522 & 36.876 & 1.7217\% & & 38.046 & 1.3965\% \\
5.3 & 40.530 & 39.347 & 2.9188\% & & 40.579 & 0.1209\% \\
5.8 & 43.331 & 41.743 & 3.6648\% & & 43.085 & 0.5677\% \\
6.3 & 46.116 & 44.064 & 4.4496\% & & 45.566 & 1.1926\% \\
6.8 & 48.916 & 46.314 & 5.3193\% & & 48.018 & 1.8358\% \\
7.3 & 51.701 & 48.493 & 6.2049\% & & 50.444 & 2.4313\% \\
7.8 & 53.763 & 50.604 & 5.8758\% & & 52.841 & 1.7149\% \\
8.3 & 55.825 & 52.649 & 5.6892\% & & 55.210 & 1.1017\% \\
8.8 & 57.887 & 54.629 & 5.6282\% & & 57.549 & 0.5839\% \\
9.3 & 59.938 & 56.547 & 5.6575\% & & 59.859 & 0.1318\% \\
9.8 & 62.001 & 58.405 & 5.7999\% & & 62.139 & 0.2226\% \\
10.3 & 64.051 & 60.204 & 6.0062\% & & 64.388 & 0.5261\% \\
\bottomrule
\end{tabularx}
\begin{tablenotes}
\footnotesize
\item[*] $\alpha_2 = 0.00176$, $\beta_2 = 1.04968$, $\sigma_2 = 0.00274$, $\lambda_0 = 0.00207$.
\end{tablenotes}
\end{threeparttable}

\bigskip

\begin{threeparttable}
\caption{Calibration results for HSBC Bank PLC CDS spreads (uncorrelated case)}
\label{tab: hsbc2}
\begin{tabularx}{.8\textwidth}
{c >{\centering\arraybackslash}X >{\centering\arraybackslash}X >{\centering\arraybackslash}X p{.05cm} >{\centering\arraybackslash}X >{\centering\arraybackslash}X}
\toprule
& & \multicolumn{2}{c}{Mapping approximation} & &\multicolumn{2}{c}{PDE approximation} \tnote{*}\\
\cmidrule{3-4} \cmidrule{6-7}
Term (year) & Market (bps) & Model (bps) & Rel. Error & & Model (bps)& Rel. Error \\
\midrule
0.7 & 14.610 & 10.474 & 28.3094\% & & 14.117 & 3.3744\% \\
1.2 & 16.770 & 13.859 & 17.3584\% & & 16.690 & 0.4770\% \\
1.7 & 19.293 & 17.102 & 11.3565\% & & 19.242 & 0.2643\% \\
2.2 & 21.802 & 20.210 & 7.3021\% & & 21.776 & 0.1193\% \\
2.7 & 23.863 & 23.189 & 2.8245\% & & 24.289 & 1.7852\% \\
3.2 & 25.912 & 26.045 & 0.5133\% & & 26.783 & 3.3614\% \\
3.8 & 29.049 & 29.318 & 0.9260\% & & 29.749 & 2.4097\% \\
4.3 & 32.185 & 31.921 & 0.8203\% & & 32.195 & 0.0311\% \\
4.8 & 34.607 & 34.417 & 0.5490\% & & 34.620 & 0.0376\% \\
5.3 & 37.015 & 36.811 & 0.5511\% & & 37.022 & 0.0189\% \\
5.8 & 39.809 & 39.107 & 1.7634\% & & 39.401 & 1.0249\% \\
6.3 & 42.588 & 41.310 & 3.0008\% & & 41.756 & 1.9536\% \\
6.8 & 45.382 & 43.423 & 4.3167\% & & 44.086 & 2.8558\% \\
7.3 & 48.161 & 45.451 & 5.6270\% & & 46.393 & 3.6710\% \\
7.8 & 49.824 & 47.397 & 4.8711\% & & 48.674 & 2.3081\% \\
8.3 & 51.488 & 49.265 & 4.3175\% & & 50.930 & 1.0837\% \\
8.8 & 53.152 & 51.058 & 3.9396\% & & 53.160 & 0.0151\% \\
9.3 & 54.806 & 52.780 & 3.6967\% & & 55.363 & 1.0163\% \\
9.8 & 56.470 & 54.434 & 3.6055\% & & 57.541 & 1.8966\% \\
10.3 & 58.124 & 56.022 & 3.6164\% & & 59.691 & 2.6960\% \\
\bottomrule
\end{tabularx}
\begin{tablenotes}
\footnotesize
\item[*] $\alpha_2 = 0.00682$, $\beta_2 = 0.25644$, $\sigma_2 = 0.02269$, $\lambda_0 = 0.00174$.
\end{tablenotes}
\end{threeparttable}
\end{table}

\begin{table}[H]
\centering
\footnotesize
\captionsetup{font=footnotesize, skip = 5pt}
\begin{minipage}{.48\textwidth}
\centering
\begin{threeparttable}
\caption{Risk-neutral survival probabilities for JP Morgan Chase \& Co (correlated case)}
\label{tab: check-jp1}
\begin{tabularx}{.9\textwidth}
{c>{\centering\arraybackslash}X>{\centering\arraybackslash}X>{\centering\arraybackslash}X}
\toprule
Term & Market& Model \tnote{*} & Rel. Error \\
(year) & & &  \\
\midrule
0.7  & 0.99805  & 0.99810  & 0.0049\% \\ 
1.2  & 0.99605  & 0.99620  & 0.0146\% \\ 
1.7  & 0.99380  & 0.99384  & 0.0038\% \\ 
2.2  & 0.99125  & 0.99103  & 0.0220\% \\ 
2.7  & 0.98810  & 0.98778  & 0.0323\% \\ 
3.2  & 0.98455  & 0.98409  & 0.0467\% \\ 
3.8  & 0.98045  & 0.97909  & 0.1391\% \\ 
4.3  & 0.97595  & 0.97445  & 0.1541\% \\ 
4.8  & 0.97080  & 0.96938  & 0.1458\% \\ 
5.3  & 0.96520  & 0.96391  & 0.1338\% \\ 
5.8  & 0.95930  & 0.95803  & 0.1327\% \\ 
6.3  & 0.95305  & 0.95175  & 0.1368\% \\ 
6.8  & 0.94640  & 0.94508  & 0.1399\% \\ 
7.3  & 0.93935  & 0.93802  & 0.1412\% \\ 
7.8  & 0.93275  & 0.93060  & 0.2303\% \\ 
8.3  & 0.92590  & 0.92282  & 0.3329\% \\ 
8.8  & 0.91880  & 0.91468  & 0.4481\% \\ 
9.3  & 0.91145  & 0.90621  & 0.5751\% \\ 
9.8  & 0.90385  & 0.89740  & 0.7133\% \\ 
10.3 & 0.89600  & 0.88828  & 0.8615\% \\ 
\bottomrule
\end{tabularx}
\end{threeparttable}
\end{minipage}%
\hfill
\begin{minipage}{.48\textwidth}
\centering
\begin{threeparttable}
\caption{Risk-neutral survival probabilities for HSBC Bank PLC (correlated case)}
\label{tab: check-hsbc1}
\begin{tabularx}{.9\textwidth}
{c>{\centering\arraybackslash}X>{\centering\arraybackslash}X>{\centering\arraybackslash}X}
\toprule
Term & Market& Model \tnote{*} & Rel. Error\\
(year) & & & \\
\midrule
0.7  & 0.99830  & 0.99835  & 0.0048\% \\ 
1.2  & 0.99645  & 0.99666  & 0.0209\% \\ 
1.7  & 0.99445  & 0.99455  & 0.0096\% \\ 
2.2  & 0.99200  & 0.99202  & 0.0018\% \\ 
2.7  & 0.98920  & 0.98908  & 0.0126\% \\ 
3.2  & 0.98615  & 0.98572  & 0.0432\% \\ 
3.8  & 0.98205  & 0.98117  & 0.0897\% \\ 
4.3  & 0.97745  & 0.97694  & 0.0525\% \\ 
4.8  & 0.97295  & 0.97231  & 0.0655\% \\ 
5.3  & 0.96810  & 0.96731  & 0.0820\% \\ 
5.8  & 0.96250  & 0.96192  & 0.0602\% \\ 
6.3  & 0.95650  & 0.95617  & 0.0349\% \\ 
6.8  & 0.95005  & 0.95005  & 0.0000\% \\ 
7.3  & 0.94325  & 0.94358  & 0.0349\% \\ 
7.8  & 0.93740  & 0.93676  & 0.0678\% \\ 
8.3  & 0.93130  & 0.92961  & 0.1811\% \\ 
8.8  & 0.92500  & 0.92214  & 0.3097\% \\ 
9.3  & 0.91860  & 0.91434  & 0.4639\% \\ 
9.8  & 0.91180  & 0.90624  & 0.6101\% \\ 
10.3 & 0.90495  & 0.89784  & 0.7860\% \\ 
\bottomrule
\end{tabularx}
\end{threeparttable}
\end{minipage}

\bigskip

\begin{minipage}{.48\textwidth}
\centering
\begin{threeparttable}
\caption{Risk-neutral survival probabilities for JP Morgan Chase \& Co (uncorrelated case)}
\label{tab: check-jp2}
\begin{tabularx}{.9\textwidth}
{c>{\centering\arraybackslash}X>{\centering\arraybackslash}X>{\centering\arraybackslash}X}
\toprule
Term & Market& Model \tnote{*} & Rel. Error\\
(year) & & & \\
\midrule
0.7  & 0.99805  & 0.99810  & 0.0050\% \\ 
1.2  & 0.99605  & 0.99620  & 0.0147\% \\ 
1.7  & 0.99380  & 0.99384  & 0.0038\% \\ 
2.2  & 0.99125  & 0.99103  & 0.0221\% \\ 
2.7  & 0.98810  & 0.98778  & 0.0325\% \\ 
3.2  & 0.98455  & 0.98409  & 0.0470\% \\ 
3.8  & 0.98045  & 0.97908  & 0.1394\% \\ 
4.3  & 0.97595  & 0.97444  & 0.1545\% \\ 
4.8  & 0.97080  & 0.96938  & 0.1463\% \\ 
5.3  & 0.96520  & 0.96390  & 0.1342\% \\ 
5.8  & 0.95930  & 0.95802  & 0.1330\% \\ 
6.3  & 0.95305  & 0.95174  & 0.1369\% \\ 
6.8  & 0.94640  & 0.94508  & 0.1398\% \\ 
7.3  & 0.93935  & 0.93803  & 0.1405\% \\ 
7.8  & 0.93275  & 0.93061  & 0.2292\% \\ 
8.3  & 0.92590  & 0.92283  & 0.3311\% \\ 
8.8  & 0.91880  & 0.91471  & 0.4455\% \\ 
9.3  & 0.91145  & 0.90624  & 0.5715\% \\ 
9.8  & 0.90385  & 0.89745  & 0.7083\% \\ 
10.3 & 0.89600  & 0.88834  & 0.8551\% \\ 
\bottomrule
\end{tabularx}
\end{threeparttable}
\end{minipage}%
\hfill
\begin{minipage}{.48\textwidth}
\centering
\begin{threeparttable}
\caption{Risk-neutral survival probabilities for HSBC Bank PLC (uncorrelated case)}
\label{tab: check-hsbc2}
\begin{tabularx}{.9\textwidth}
{c>{\centering\arraybackslash}X>{\centering\arraybackslash}X>{\centering\arraybackslash}X}
\toprule
Term & Market& Model \tnote{*} & Rel. Error \\
(year) & & & \\
\midrule
0.7  & 0.99830  & 0.99836  & 0.0058\% \\ 
1.2  & 0.99645  & 0.99667  & 0.0222\% \\ 
1.7  & 0.99445  & 0.99456  & 0.0111\% \\ 
2.2  & 0.99200  & 0.99203  & 0.0030\% \\ 
2.7  & 0.98920  & 0.98909  & 0.0114\% \\ 
3.2  & 0.98615  & 0.98574  & 0.0418\% \\ 
3.8  & 0.98205  & 0.98119  & 0.0874\% \\ 
4.3  & 0.97745  & 0.97698  & 0.0486\% \\ 
4.8  & 0.97295  & 0.97238  & 0.0590\% \\ 
5.3  & 0.96810  & 0.96740  & 0.0718\% \\ 
5.8  & 0.96250  & 0.96207  & 0.0446\% \\ 
6.3  & 0.95650  & 0.95638  & 0.0120\% \\ 
6.8  & 0.95005  & 0.95036  & 0.0322\% \\ 
7.3  & 0.94325  & 0.94400  & 0.0792\% \\ 
7.8  & 0.93740  & 0.93732  & 0.0088\% \\ 
8.3  & 0.93130  & 0.93033  & 0.1042\% \\ 
8.8  & 0.92500  & 0.92305  & 0.2111\% \\ 
9.3  & 0.91860  & 0.91548  & 0.3398\% \\ 
9.8  & 0.91180  & 0.90764  & 0.4562\% \\ 
10.3 & 0.90495  & 0.89954  & 0.5976\% \\ 
\bottomrule
\end{tabularx}
\end{threeparttable}
\end{minipage}

\begin{threeparttable}
\begin{tablenotes}
\item[*] We apply the \textit{PDE approximation} method here.
\end{tablenotes}
\end{threeparttable}

\end{table}

\begin{table}[htbp]
\centering
\begin{threeparttable}
\captionsetup{font=footnotesize, skip = 5pt}
\caption{Computational time of the calibrations}
\label{tab: time2}
\footnotesize
\begin{tabular}{lcc}
\toprule
Reference entity & Correlated case & Uncorrelated case\\
& (second) & (second)\\
\midrule
JP Morgan Chase \& Co & 214.57 & 180.78\\
HSBC Bank PLC & 224.31 & 178.16 \\
\bottomrule
\end{tabular}
\end{threeparttable}
\end{table}

\subsection{Negative interest rates}
In this section, we aim to compare our CDS spread approximation in the SSRD model with the approximation of \citet{di_francesco_cds_2019}, which employs the same asymptotic approximation method of \citet{lorig_analytical_2015}, but in the extended jump-to-default constant elasticity of variance (JDCEV) model (\citet{carr_jump_2006}). In their model setup, they assumed that the interest rate follows a Vasicek process allowing for negative values. We use the market data as reported in \citet{di_francesco_cds_2019} for performing our comparative calibration study. From Table \ref{tab: zcb1}, it is apparent that the CIR model for the interest rate as assumed in the SSRD model offers a much better fit than the Vasicek model used in the extended JDCEV model. The volatility parameter $\widehat{\sigma}_1$ (\textit{Step 2} of Section \ref{sec: calibration approach}) is computed to be  $0.11836$ for longest 6-year maturity of CDS spreads. 

\begin{table}[htbp]
\centering
\footnotesize
\captionsetup{font=footnotesize, skip = 5pt}
\begin{threeparttable}
\caption{Calibration results for ZCB prices (LIBOR)}
\label{tab: zcb1}
\begin{tabularx}{.55\textwidth}
{c>{\centering\arraybackslash}X>{\centering\arraybackslash}X>{\centering\arraybackslash}X>{\centering\arraybackslash}X>{\centering\arraybackslash}X}
\toprule
Term & Market & Vasicek & Rel. & CIR & Rel. \\
(year) &  &  & Error &  & Error \\
\midrule
1 & 1.00229 & 1.00620 & 0.3900\% & 1.00653 & 0.4232\% \\
2 & 1.00371 & 1.00731 & 0.3587\% & 1.00864 & 0.4910\% \\
3 & 1.00333 & 1.00408 & 0.0750\% & 1.00703 & 0.3684\% \\
4 & 1.00099 & 0.99722 & 0.3766\% & 1.00233 & 0.1335\% \\
5 & 0.99584 & 0.98737 & 0.8500\% & 0.99509 & 0.0752\% \\
6 & 0.98787 & 0.97512 & 1.2906\% & 0.98578 & 0.2109\% \\
7 & 0.97701 & 0.96097 & 1.6409\% & 0.97482 & 0.2235\% \\
8 & 0.96360 & 0.94538 & 1.8901\% & 0.96254 & 0.1094\% \\
9 & 0.94837 & 0.92873 & 2.0709\% & 0.94923 & 0.0909\% \\
10 & 0.93383 & 0.91134 & 2.4081\% & 0.93513 & 0.1397\% \\
\bottomrule
\end{tabularx}
\begin{tablenotes}
\item $\alpha_1 = 0.18083$, $\beta_1 = 0.02021$, $\sigma_1 = 0.00193$, $r_0 = -0.009$.
\end{tablenotes}
\end{threeparttable}
\end{table}

In this particular calibration study, the following choice of weights in the final optimisation step (\textit{Step 3} in Section \ref{sec: calibration approach}) provides the best results. 
\eqstar{
\omega_i = \frac{\frac{1}{T_i}}{\sum_{i=1}^N {1\over T_i}}.
}
The idea here is to assign more weight to the CDS spread values corresponding to short-term maturities as they are susceptible to give large calibration errors. The calibration results from our approximation method are presented in Table \ref{tab: bnp1} - Table \ref{tab: ubs2}, along with the estimates reported in \citet{di_francesco_cds_2019}. It is evident that the estimates using our approximation formula \eqref{eq: cds spread0} in the transformed SSRD model \eqref{eq: ssrd t} closely align with the market data, for both correlated and uncorrelated cases. The relative errors are notably small, except for the very short maturity values, such as 1-year maturity, and non-liquid terms, such as the CDS with 4-year maturity, which are attributed to market incompleteness, as argued in \citet{di_francesco_cds_2019}. In addition to the reference entities BNP Paribas and UBS AG, we also present the calibration results for four other distinct entities in Table \ref{tab: caixa1} - Table \ref{tab: medio1}. Our approximation formula also yields excellent results for those entities as well.

\begin{table}[H]
\centering
\footnotesize
\captionsetup{font = footnotesize, skip = 5pt}
\begin{threeparttable}
\caption{Calibration results for BNP Paribas CDS spreads (correlated case)}
\label{tab: bnp1}
\begin{tabularx}{.65\textwidth}
{c>{\centering\arraybackslash}X>{\centering\arraybackslash}X>{\centering\arraybackslash}X>{\centering\arraybackslash}X>{\centering\arraybackslash}X}
\toprule
Term & Market & JDCEV & Rel. Error & SSRD & Rel. Error\\
(year) & (bps) & (bps) &  & (bps) &  \\
\midrule
1.0 & 34.615 & 34.007 & 1.7556\% & 33.491 & 3.2471\% \\
1.5 & 39.876 & 39.668 & 0.5211\% & 39.754 & 0.3054\% \\
2.0 & 45.112 & 45.481 & 0.8197\% & 46.002 & 1.9740\% \\
2.5 & 49.994 & 51.444 & 2.9006\% & 52.228 & 4.4696\% \\
3.0 & 56.110 & 57.547 & 2.5601\% & 58.428 & 4.1312\% \\
3.5 & 64.573 & 63.783 & 1.2233\% & 64.594 & 0.0331\% \\
4.0 & 72.590 & 70.145 & 3.3684\% & 70.722 & 2.5734\% \\
4.5 & 77.590 & 76.626 & 1.2426\% & 76.806 & 1.0104\% \\
5.0 & 82.270 & 83.218 & 1.1528\% & 82.843 & 0.6965\% \\
5.5 & 89.146 & 89.916 & 0.8639\% & 88.826 & 0.3584\% \\
6.0 & 96.705 & 96.711 & 0.0057\% & 94.753 & 2.0185\% \\
\bottomrule
\end{tabularx}
\begin{tablenotes}
\item $\alpha_2 = 0.00150$, $\beta_2 = 2.79211$, $\sigma_2 = 0.00033$, $\lambda_0 = 0.00350$, $\rho = 0.06969$.
\end{tablenotes}
\end{threeparttable}
\end{table}
\begin{table}[H]
\centering
\footnotesize
\captionsetup{font=footnotesize, skip = 5pt}
\begin{threeparttable}
\caption{Calibration results for UBS AG CDS spreads (correlated case)}
\label{tab: ubs1}
\begin{tabularx}{.65\textwidth}
{c>{\centering\arraybackslash}X>{\centering\arraybackslash}X>{\centering\arraybackslash}X>{\centering\arraybackslash}X>{\centering\arraybackslash}X}
\toprule
Term & Market & JDCEV & Rel. Error & SSRD & Rel. Error\\
(year) & (bps) & (bps) &  & (bps) &  \\
\midrule
1.0 & 25.720 & 25.962 & 0.9417\% & 25.672 & 0.1866\% \\
1.5 & 30.263 & 30.361 & 0.3225\% & 30.296 & 0.1090\% \\
2.0 & 35.105 & 34.787 & 0.9059\% & 34.900 & 0.5840\% \\
2.5 & 39.692 & 39.242 & 1.1332\% & 39.482 & 0.5291\% \\
3.0 & 43.970 & 43.726 & 0.5545\% & 44.037 & 0.1524\% \\
3.5 & 48.062 & 48.239 & 0.3674\% & 48.562 & 1.0403\% \\
4.0 & 52.300 & 52.779 & 0.9149\% & 53.054 & 1.4417\% \\
4.5 & 56.965 & 57.345 & 0.6662\% & 57.510 & 0.9567\% \\
5.0 & 61.910 & 61.935 & 0.0396\% & 61.928 & 0.0291\% \\
5.5 & 66.841 & 66.546 & 0.4412\% & 66.305 & 0.8019\% \\
6.0 & 71.285 & 71.176 & 0.1526\% & 70.639 & 0.9062\% \\
\bottomrule
\end{tabularx}
\begin{tablenotes}
\item $\alpha_2 = 0.00803$, $\beta_2 = 0.38923$, $\sigma_2 = 0.00176$, $\lambda_0 = 0.00274$, $\rho = 0.96$.
\end{tablenotes}
\end{threeparttable}
\end{table}

\begin{table}[H]
\centering
\footnotesize
\captionsetup{font=footnotesize, skip=5pt}
\begin{threeparttable}
\caption{Calibration results for BNP Paribas CDS spreads (uncorrelated case)}
\label{tab: bnp2}
\begin{tabularx}{.65\textwidth}
{c>{\centering\arraybackslash}X>{\centering\arraybackslash}X>{\centering\arraybackslash}X>{\centering\arraybackslash}X>{\centering\arraybackslash}X}
\toprule
Term & Market & JDCEV & Rel. Error & SSRD & Rel. Error\\
(year) & (bps) & (bps) &  & (bps) &  \\
\midrule
1.0 & 34.615 & 34.054 & 1.6221\% & 33.474 & 3.2963\% \\
1.5 & 39.876 & 39.674 & 0.5071\% & 39.739 & 0.3431\% \\
2.0 & 45.112 & 45.464 & 0.7803\% & 45.988 & 1.9430\% \\
2.5 & 49.994 & 51.415 & 2.8424\% & 52.216 & 4.4456\% \\
3.0 & 56.110 & 57.517 & 2.5076\% & 58.417 & 4.1116\% \\
3.5 & 64.573 & 63.761 & 1.2566\% & 64.584 & 0.0177\% \\
4.0 & 72.590 & 70.136 & 3.3801\% & 70.714 & 2.5844\% \\
4.5 & 77.590 & 76.632 & 1.2351\% & 76.799 & 1.0195\% \\
5.0 & 82.270 & 83.236 & 1.1738\% & 82.837 & 0.6892\% \\
5.5 & 89.146 & 89.937 & 0.8874\% & 88.822 & 0.3629\% \\
6.0 & 96.705 & 96.722 & 0.0178\% & 94.750 & 2.0216\% \\
\bottomrule
\end{tabularx}
\begin{tablenotes}
\item $\alpha_2 = 0.00361$, $\beta_2 = 1.16567$, $\sigma_2 = 0.00097$, $\lambda_0 = 0.0035$.
\end{tablenotes}
\end{threeparttable}
\end{table}
\begin{table}[H]
\centering
\footnotesize
\captionsetup{font=footnotesize, skip = 5pt}
\begin{threeparttable}
\caption{Calibration results for UBS AG CDS spreads (uncorrelated case)}
\label{tab: ubs2}
\begin{tabularx}{.65\textwidth}
{c>{\centering\arraybackslash}X>{\centering\arraybackslash}X>{\centering\arraybackslash}X>{\centering\arraybackslash}X>{\centering\arraybackslash}X}
\toprule
Term & Market & JDCEV & Rel. Error & SSRD & Rel. Error\\
(year) & (bps) & (bps) &  & (bps) &  \\
\midrule
1.0 & 25.720 & 25.914 & 0.7543\% & 25.700 & 0.0778\% \\
1.5 & 30.263 & 30.331 & 0.2254\% & 30.321 & 0.1917\% \\
2.0 & 35.105 & 34.781 & 0.9218\% & 34.920 & 0.5270\% \\
2.5 & 39.692 & 39.261 & 1.0869\% & 39.493 & 0.5014\% \\
3.0 & 43.970 & 43.765 & 0.4674\% & 44.037 & 0.1524\% \\
3.5 & 48.062 & 48.289 & 0.4723\% & 48.547 & 1.0091\% \\
4.0 & 52.300 & 52.830 & 1.0132\% & 53.022 & 1.3805\% \\
4.5 & 56.965 & 57.383 & 0.7343\% & 57.458 & 0.8654\% \\
5.0 & 61.910 & 61.945 & 0.0569\% & 61.852 & 0.0937\% \\
5.5 & 66.841 & 66.512 & 0.4930\% & 66.203 & 0.9545\% \\
6.0 & 71.285 & 71.078 & 0.2898\% & 70.507 & 1.0914\% \\
\bottomrule
\end{tabularx}
\begin{tablenotes}
\item $\alpha_2 = 0.01021$, $\beta_2 = 0.30701$, $\sigma_2 = 0.00601$, $\lambda_0 = 0.00274$.
\end{tablenotes}
\end{threeparttable}
\end{table}

\begin{table}[htbp]
\centering
\footnotesize
\captionsetup{font=footnotesize, skip = 5pt}
\begin{threeparttable}
\caption{Computational time of the calibrations}
\label{tab: time1}
\begin{tabular}{lcc}
\toprule
Reference entity & Correlated case & Uncorrelated case\\
& (second) & (second)\\
\midrule
BNP Paribas & 210.02 & 167.96\\
UBS AG & 243.90 & 190.16\\
\bottomrule
\end{tabular}
\end{threeparttable}
\end{table}

Moreover, we compute the estimate of the risk-neutral survival probability using \eqref{eq: survival probability expansion} and compare it with the market data. The results are reported in Table \ref{tab: check-bnp1} - Table \ref{tab: check-ubs2}. Once again, it is evident that our approximation formula closely aligns with the survival probability inferred from the market CDS spreads.

\begin{table}[H]
\centering
\footnotesize
\captionsetup{font=footnotesize, skip = 5pt}
\begin{threeparttable}
\caption{Risk-neutral survival probabilities for BNP Paribas (correlated case)}
\label{tab: check-bnp1}
\begin{tabularx}{.7\textwidth}
{c>{\centering\arraybackslash}X>{\centering\arraybackslash}X>{\centering\arraybackslash}X>{\centering\arraybackslash}X>{\centering\arraybackslash}X}
\toprule
Term & Market & JDCEV & Rel. Error & SSRD & Rel. Error\\
(year) &  &  &  &  &  \\
\midrule
1  & 0.99425  & 0.99435  & 0.0094\%  & 0.99443  & 0.0174\% \\ 
2  & 0.98508  & 0.98494  & 0.0139\%  & 0.98476  & 0.0323\% \\ 
3  & 0.97230  & 0.97154  & 0.0780\%  & 0.97113  & 0.1210\% \\ 
4  & 0.95254  & 0.95398  & 0.1510\%  & 0.95370  & 0.1219\% \\ 
5  & 0.93328  & 0.93212  & 0.1254\%  & 0.93270  & 0.0628\% \\ 
6  & 0.90887  & 0.90590  & 0.3272\%  & 0.90838  & 0.0545\% \\ 
\bottomrule
\end{tabularx}
\end{threeparttable}
\end{table}
\begin{table}[H]
\centering
\footnotesize
\captionsetup{font=footnotesize, skip = 5pt}
\begin{threeparttable}
\caption{Risk-neutral survival probabilities for UBS AG (correlated case)}
\label{tab: check-ubs1}
\begin{tabularx}{.7\textwidth}
{c>{\centering\arraybackslash}X>{\centering\arraybackslash}X>{\centering\arraybackslash}X>{\centering\arraybackslash}X>{\centering\arraybackslash}X}
\toprule
Term & Market & JDCEV & Rel. Error & SSRD & Rel. Error\\
(year) &  &  &  &  &  \\
\midrule
1  & 0.99572  & 0.99568  & 0.0042\%  & 0.99572  & 0.0000\% \\ 
2  & 0.98837  & 0.98847  & 0.0105\%  & 0.98842  & 0.0053\% \\ 
3  & 0.97823  & 0.97835  & 0.0123\%  & 0.97817  & 0.0060\% \\ 
4  & 0.96564  & 0.96532  & 0.0331\%  & 0.96511  & 0.0557\% \\ 
5  & 0.94944  & 0.94940  & 0.0040\%  & 0.94936  & 0.0086\% \\ 
6  & 0.93056  & 0.93062  & 0.0063\%  & 0.93108  & 0.0560\% \\ 
\bottomrule
\end{tabularx}
\end{threeparttable}
\end{table}

\begin{table}[H]
\centering
\footnotesize
\captionsetup{font=footnotesize, skip = 5pt}
\begin{threeparttable}
\caption{Risk-neutral survival probabilities for BNP Paribas (uncorrelated case)}
\label{tab: check-bnp2}
\begin{tabularx}{.7\textwidth}
{c>{\centering\arraybackslash}X>{\centering\arraybackslash}X>{\centering\arraybackslash}X>{\centering\arraybackslash}X>{\centering\arraybackslash}X}
\toprule
Term & Market & JDCEV & Rel. Error & SSRD & Rel. Error\\
(year) &\\
\midrule
1  & 0.99425  & 0.99434  & 0.0086\%  & 0.99443  & 0.0176\% \\ 
2  & 0.98508  & 0.98495  & 0.0131\%  & 0.98476  & 0.0319\% \\ 
3  & 0.97230  & 0.97157  & 0.0753\%  & 0.97113  & 0.1204\% \\ 
4  & 0.95254  & 0.95402  & 0.1551\%  & 0.95370  & 0.1224\% \\ 
5  & 0.93328  & 0.93218  & 0.1189\%  & 0.93270  & 0.0624\% \\ 
6  & 0.90887  & 0.90602  & 0.3135\%  & 0.90838  & 0.0541\% \\ 
\bottomrule
\end{tabularx}
\end{threeparttable}
\end{table}
\begin{table}[H]
\centering
\footnotesize
\captionsetup{font=footnotesize, skip = 5pt}
\begin{threeparttable}
\caption{Risk-neutral survival probabilities for UBS AG (uncorrelated case)}
\label{tab: check-ubs2}
\begin{tabularx}{.7\textwidth}
{c>{\centering\arraybackslash}X>{\centering\arraybackslash}X>{\centering\arraybackslash}X>{\centering\arraybackslash}X>{\centering\arraybackslash}X}
\toprule
Term & Market & JDCEV & Rel. Error & SSRD & Rel. Error\\
(year) &\\
\midrule
1  & 0.99572  & 0.99569  & 0.0034\%  & 0.99572  & 0.0004\% \\ 
2  & 0.98837  & 0.98847  & 0.0107\%  & 0.98841  & 0.0046\% \\ 
3  & 0.97823  & 0.97834  & 0.0104\%  & 0.97818  & 0.0056\% \\ 
4  & 0.96564  & 0.96529  & 0.0364\%  & 0.96514  & 0.0524\% \\ 
5  & 0.94944  & 0.94939  & 0.0048\%  & 0.94944  & 0.0003\% \\ 
6  & 0.93056  & 0.93071  & 0.0157\%  & 0.93125  & 0.0741\% \\ 
\bottomrule
\end{tabularx}
\end{threeparttable}
\end{table}

\section{Conclusion}

In this work, we provide a closed-form approximation formula for the CDS spread under the SSRD model framework. Specifically, we utilise Taylor's theorem to derive approximations for the two crucial components within the general CDS spread formula, which are solutions to Cauchy problems. Most notably, we derive our approximation without assuming uncorrelated interest rate and default intensity, as required by \citet{brigo_credit_2005} for their calibration procedure. With several numerical studies using different market data on CDS spread, we demonstrate the efficiency of our calibration procedure and accuracy of our CDS spread approximation, while comparing it with other existing approximations. 

\printbibliography

\begin{appendix}
\section*{Appendix}

\renewcommand{\thesubsection}{A.\arabic{subsection}}

\setcounter{equation}{0}
\renewcommand{\theequation}{\thesubsection.\arabic{equation}}

\newtheorem{theoremA}{Theorem}[subsection]
\newtheorem{definitionA}[theoremA]{Definition}
\newtheorem{lemmaA}[theoremA]{Lemma}
\newtheorem{corollaryA}[theoremA]{Corollary}
\newtheorem{remarkA}[theoremA]{Remark}
\newtheorem{propositionA}[theoremA]{Proposition}

\renewcommand{\thetheoremA}{\thesubsection.\arabic{theoremA}}
\renewcommand{\thedefinitionA}{\thesubsection.\arabic{theoremA}}
\renewcommand{\thelemmaA}{\thesubsection.\arabic{theoremA}}
\renewcommand{\thecorollaryA}{\thesubsection.\arabic{theoremA}}
\renewcommand{\thepropositionA}{\thesubsection.\arabic{theoremA}}

\numberwithin{table}{subsection}

\subsection{The approximation formula}
\label{sec: 2nd expression}

Within the framework established in Section \ref{sec: model setting}, \ref{sec: asymptotic approximation technique} and \ref{sec: approximation formula}, we derive the the $N$-th order approximation formula for CDS spread at time $0$, as given in \eqref{eq: cds approximation}. However, we restrict the display of the explicit expression for the approximation formula to the second order, as various numerical experiments indicate that the second-order approximation is sufficiently accurate.

From Proposition \ref{pro: cds n}, we obtain the second-order approximation formula of CDS spread at $t=0$, given by
\eqstar{
R_2 = \frac{(1-\zeta) \int_0^T \ee^{-\alpha_2 s} \sum_{n=0}^2 h_n(0, x, y, s) \dd s}{ \int_0^T \ee^{-\alpha_2 s} \sum_{n=0}^2 h_n(0, x, y, s) (s - t_{N(s)-1}) \dd s 
+ \sum_{i=1}^M (t_i - t_{i-1}) \sum_{n=0}^2 v_n(0, x, y, t_i)},
}
with the zeroth term $v_0$ and $h_0$ specified in \eqref{eq: v0} and \eqref{eq: h0}, respectively. To simplify the derivation process, we introduce the following useful functions:
\eqstar{
\Psi_0 (\alpha, i, j, t_1, t_2) &:= \int_{t_1}^{t_2} \ee^{\alpha s} \bar{x}(t, s)^{i\over 2} \bar{y}(t, s)^{j\over 2} \dd s,\\
\Psi_1 (\alpha, i, j, k, t_1, t_2) &:= \int_{t_1}^{t_2} \ee^{\alpha s} \bar{x}(t, s)^{i\over 2} \bar{y}(t, s)^{j\over 2} \mathbf{C}_{1,1} (t, s)^k \dd s,\\
\Psi_2 (\alpha, i, j, k, t_1, t_2) &:= \int_{t_1}^{t_2} \ee^{\alpha s} \bar{x}(t, s)^{i\over 2} \bar{y}(t, s)^{j\over 2} \mathbf{C}_{2,2} (t, s)^k \dd s,\\
\Psi_3 (\alpha, i, j, k, t_1, t_2) &:= \int_{t_1}^{t_2} \ee^{\alpha s} \bar{x}(t, s)^{i\over 2} \bar{y}(t, s)^{j\over 2} \mathbf{C}_{1,2} (t, s)^k \dd s, \\
&= \int_{t_1}^{t_2} \ee^{\alpha s} \bar{x}(t, s)^{i\over 2} \bar{y}(t, s)^{j\over 2} \mathbf{C}_{2,1} (t, s)^k \dd s,\\
\Phi_1 (\alpha, i, j, t_1, t_2) &:= \int_{t_1}^{t_2} \ee^{\alpha s} \bar{x}(t, s)^{i\over 2} \bar{y}(t, s)^{j\over 2} \mathbf{C}_{1,1} (t, s) \mathbf{C}_{1,2} (t, s) \dd s,\\
\Phi_2 (\alpha, i, j, t_1, t_2) &:= \int_{t_1}^{t_2} \ee^{\alpha s} \bar{x}(t, s)^{i\over 2} \bar{y}(t, s)^{j\over 2} \mathbf{C}_{2,2} (t, s) \mathbf{C}_{1,2} (t, s) \dd s ,\\
\Phi_3 (\alpha, i, j, t_1, t_2) &:= \int_{t_1}^{t_2} \ee^{\alpha s} \bar{x}(t, s)^{i\over 2} \bar{y}(t, s)^{j\over 2} \mathbf{C}_{1,1} (t, s) \mathbf{C}_{2,2} (t, s) \dd s,
}
for any $\alpha \in \Rb$, $i, j, k \in \Nb_0$, and $t\le t_1\le t_2 \le T$, where $\bar{x}$ and $\bar{y}$ are as defined in \eqref{eq: term bar}. Moreover, for all coefficient functions in the operator $\Ac(t)$ in \eqref{eq: term AA}, we set, taking $a(t, x, y)$ as an example,
\eqstar{
a_{i,j} (t) := \partial_x^i \partial_y^j a(t, \bar{x}, \bar{y}), && i, j \in \Nb_0
}
Thus, the zeroth order term of the expansion for $\Ac(t)$ in \eqref{eq: term AA} can be rewritten as follows
\eqstar{
\Ac_0 (s) = a_{0,0}(s) \partial_{xx}+b_{0,0}(s) \partial_{yy}+c_{0,0}(s) \partial_{xy}+\kappa_{0,0}(s) \partial_{x}+k_{0,0}(s) \partial_{y}+\gamma_{0,0}(s),
}
where
\eqstar{
a_{0,0}(s) &= {1\over 2}\sigma_1^2\ee^{\alpha_1 s}\bar{x} (s), 
&&\kappa_{0,0}(s) = \alpha_1\beta_1\ee^{\alpha_1 s},
&& c_{0,0}(s) = \widehat{\rho} \bar{x}^{1\over 2} \bar{y}^{1\over 2} \ee^{\bar{\alpha} s},\\
b_{0,0}(s) &= {1\over 2}\sigma_2^2\ee^{\alpha_2 s}\bar{y} (s),
&& k_{0,0}(s) = \alpha_2\beta_2\ee^{\alpha_2 s},
&& \gamma_{0,0}(s) = -(\ee^{-\alpha_1 t}\bar{x} (s)+\ee^{-\alpha_2 t}\bar{y} (s)),
}
with
\eqstar{
\widehat{\rho}:= \rho\sigma_1\sigma_2, && \bar{\alpha}:=\frac{\alpha_1+\alpha_2}{2}.
}
Applying \eqref{eq: term Gamma} and \eqref{eq: term cm}, we get the covariance matrix $\mathbf{C}(t,T)$ and the mean vector $z+\mathbf{m}(t,T)$, expressed as follows
\eqstar{
\mathbf{C}(t,T) :=\begin{pmatrix}
\mathbf{C}_{1,1}(t,T) & \mathbf{C}_{1,2}(t,T) \\
\mathbf{C}_{2,1}(t,T)  & \mathbf{C}_{2,2}(t,T) 
\end{pmatrix}, &&
\mathbf{m}(t,T) := \binom{\alpha_1 \beta_1 \psi(\alpha_1,t,T)}{\alpha_2 \beta_2 \psi(\alpha_2,t,T)},
}
where the elements of $\mathbf{C}(t,T)$ are given by
\eqstar{
\mathbf{C}_{1,1}(t,T) &:= \sigma_1^2 \left(\bar{x}_{\mathrm{fixed}}\psi(\alpha_1,t,T) + \alpha_1\beta_1 \Theta(\alpha_1, \alpha_1, t, T)\right),\\
\mathbf{C}_{2,2}(t,T) &:= \sigma_2^2 \left(\bar{y}_{\mathrm{fixed}}\psi(\alpha_2,t,T) + \alpha_2\beta_2 \Theta(\alpha_2, \alpha_2, t, T)\right),\\
\mathbf{C}_{1,2}(t,T) &= \mathbf{C}_{2,1}(t,T) := \mathbf{C}_{2,1}(t,T) := \widehat{\rho} \Psi_0(\bar{\alpha}, 1, 1, t, T).
}
Theorem \ref{thm: un} yields the first term $v_1$, for $t \in [0,T)$, as follows
\eqlnostar{eq: v1}{
v_1(t,x,y,T) = \Lc_1 (t, T) v_0(t,x,y,T) = \int_t^T \dd s_1 \Gc_1(t,s_1) v_0(t,x,y,T),
}
where
\eqlnostar{eq: G1}{
\Gc_1(t,s) &= \Ac_1(s, \Mc(t, s)),
}
with
\eqlnostar{eq: M}{
\Mc (t,s) &= \begin{pmatrix} 
\Mc_1 (t,s)\\
\Mc_2 (t,s)
\end{pmatrix}
:= 
\begin{pmatrix}
x + \alpha_x \beta_x \psi(\alpha_x,t,s) + \mathbf{C}_{1,1}(t,s) \partial_x + \mathbf{C}_{1,2}(t,s) \partial_y \\
y + \alpha_y \beta_y \psi(\alpha_y,t,s) + \mathbf{C}_{2,2}(t,s) \partial_y + \mathbf{C}_{2,1}(t,s) \partial_x
\end{pmatrix}.
}
Inserting \eqref{eq: M} into \eqref{eq: G1}, we obtain
\eqlnostar{eq: G1e}{
\Gc_1 (t, s) &= (\Mc_1 (t, s) - \bar{x} (t, s)) \Ac_{1,0} (s) + (\Mc_2 (t, s) - \bar{y} (t, s)) \Ac_{0,1} (s),\nonumber \\
&= \xt (t, s) \Ac_{1,0} (s) + \yt (t, s) \Ac_{0,1} (s),
}
where $\xt$ and $\yt$ are defined as follows
\eqstar{
\xt (t, s) := \mathbf{C}_{1,1}(t,s) \partial_x + \mathbf{C}_{1,2}(t,s) \partial_y,
&&\yt (t, s) := \mathbf{C}_{2,2}(t,s) \partial_y + \mathbf{C}_{2,1}(t,s) \partial_x,
}
and $\Ac_{1,0} (s)$ and $\Ac_{0,1} (s)$ are components of the first order term of $\Ac(t)$ expansion. Specifically, $\Ac_1(s) = \Ac_{1,0} (s) + \Ac_{0,1} (s)$, where each term is defined as
\eqstar{
\Ac_{1,0} (s) &:= a_{1,0}(s) \partial_{xx} + c_{1,0}(s) \partial_{xy} +\gamma_{1,0}(s),\\
\Ac_{0,1} (s) &:= b_{0,1}(s) \partial_{yy} + c_{0,1}(s) \partial_{xy} +\gamma_{0,1}(s),
}
with 
\eqstar{
a_{1,0}(1) = {1\over 2}\sigma_1^2 \ee^{\alpha_1 s}, 
&& c_{1,0}(s) &= {1\over 2} \widehat{\rho} \bar{x}^{-{1\over 2}} (t, s) \bar{y}^{1\over 2} (t, s) \ee^{\bar{\alpha} s}, 
&& \gamma_{1,0}(s) = -\ee^{-\alpha_1 s},\\
b_{0,1}(s) = {1\over 2}\sigma_2^2 \ee^{\alpha_2 s}, 
&& c_{0,1}(s) &= {1\over 2} \widehat{\rho} \bar{x}^{1\over 2} (t, s) \bar{y}^{-{1\over 2}} (t, s) \ee^{\bar{\alpha} s}, 
&& \gamma_{0,1}(s) = -\ee^{-\alpha_2 s},
}
Therefore, the explicit expression of $\Lc_1 (t, T)$ is given by
\eqlnostar{eq: L1}{
\Lc_1 (t, T) = \sum_{i=1}^6 L_{1,0} ^{(i)} (t, T) + \sum_{i=1}^6 L_{0,1} ^{(i)} (t, T),
}
where each component is defined as follows
\eqstar{
L_{1,0} ^{(1)} (t, T) &:= {1\over 2}\sigma_1^2 \Psi_1 (\alpha_1, 0, 0, 1, t, T) \partial_x^3, 
&& L_{1,0} ^{(2)} (t, T) := {1\over 2}\sigma_1^2 \Psi_3 (\alpha_1, 0, 0, 1, t, T) \partial_x^2\partial_y,\\
L_{1,0} ^{(3)} (t, T) &:= {1\over 2} \widehat{\rho} \Psi_1 (\bar{\alpha}, -1, 1, 1, t, T) \partial_x^2\partial_y,
&& L_{1,0} ^{(4)} (t, T) := {1\over 2} \widehat{\rho} \Psi_3 (\bar{\alpha}, -1, 1, 1, t, T) \partial_x\partial_y^2,\\
L_{1,0} ^{(5)} (t, T) &:= -\Psi_1 (-\alpha_1, 0, 0, 1, t, T) \partial_x,
&& L_{1,0} ^{(6)} (t, T) := -\Psi_3 (-\alpha_1, 0, 0, 1, t, T) \partial_y,
}
and
\eqstar{
L_{0,1} ^{(1)} (t, T) &:= {1\over 2}\sigma_2^2 \Psi_2 (\alpha_2, 0, 0, 1, t, T) \partial_y^3, 
&& L_{0,1} ^{(2)} (t, T) := {1\over 2}\sigma_2^2 \Psi_3 (\alpha_2, 0, 0, 1, t, T) \partial_x\partial_y^2,\\
L_{0,1} ^{(3)} (t, T) &:= {1\over 2} \widehat{\rho} \Psi_2 (\bar{\alpha}, 1, -1, 1, t, T) \partial_x\partial_y^2,
&& L_{0,1} ^{(4)} (t, T) := {1\over 2} \widehat{\rho} \Psi_3 (\bar{\alpha}, 1, -1, 1, t, T) \partial_x^2\partial_y,\\
L_{0,1} ^{(5)} (t, T) &:= -\Psi_2 (-\alpha_2, 0, 0, 1, t, T) \partial_y,
&& L_{0,1} ^{(6)} (t, T) := -\Psi_3 (-\alpha_2, 0, 0, 1, t, T) \partial_x.
}
We observe that the partial derivatives of $v_0$ in \eqref{eq: v0} exhibit the property such that
\eqlnostar{eq: partial derivative}{
\partial_x^i \partial_y^j v_0(t,x,y,T) = (-1)^{i+j} \psi(-\alpha_1,t,T)^i \psi(-\alpha_2,t,T)^j v_0(t,x,y,T), && i,j\in\mathbb{N}_0.
}
Multiplying \eqref{eq: v0} and \eqref{eq: L1} in accordance with the principle outlined in \eqref{eq: partial derivative}, we derive the explicit expression for $v_1 (t, x, y, T)$. Analogous to \eqref{eq: v1}, the second term in the approximation of $v(t, x, y, T)$, for any $t\in [0, T)$, is given as
\eqlnostar{eq: v2}{
v_2(t,x,y,T) &= \Lc_1 (t, T) v_0(t,x,y,T)\\
&= \left(\int_t^T \dd s_1 \Gc_2(t,s_1) + \int_t^T \dd s_1 \Gc_1(t,s_1) \int_{s_1}^T \dd s_2 \Gc_1(t,s_2) \right) v_0(t,x,y,T),
}
where $\Gc_1(t,s)$ is as specified in \eqref{eq: G1}, and $\Gc_2(t,s)$ is given by
\eqstar{
\Gc_2(t,s) = {1\over 2} \xt (t, s)^2 \Ac_{2,0} (s) + \xt (t, s) \yt (t, s) \Ac_{1,1} (s) + \yt (t, s)^2 \Ac_{0,2} (s).
}
Here, $\Ac_{2,0} (s)$, $\Ac_{1,1} (s)$, and $\Ac_{0,2} (s)$ are components of $\Ac_2(s)$, which are defined as follows
\eqstar{
\Ac_{2,0} (s) &:= -{1\over 4}\widehat{\rho} \bar{x}^{-{3\over 2}}(t, s) \bar{y}^{1\over 2} (t, s) \ee^{\bar{\alpha} t} \partial_{xy},\\
\Ac_{1,1} (s) &:= {1\over 4}\widehat{\rho} \bar{x}^{-{1\over 2}} (t, s) \bar{y}^{-{1\over 2}} (t, s)\ee^{\bar{\alpha} t} \partial_{xy},\\
\Ac_{0,2} (s) &:= -{1\over 4}\widehat{\rho} \bar{x}^{1\over 2} (t, s) \bar{y}^{-{3\over 2}} (t, s) \ee^{\bar{\alpha} t} \partial_{xy}.
}
Accordingly, we divide the operator $\Lc_2 (t, T)$ into two parts, such that $\Lc_2 (t, T) = \Lc_2 ^{(1)} (t, T) + \Lc_2 ^{(2)} (t, T)$. We define the first component as follows
\eqlnostar{eq: L21}{
\Lc_2 ^{(1)} (t, T) &:= \int_t^T \dd s_1 \Gc_2(t,s_1) \nonumber\\
&= \widehat{\rho} \left(-{1\over 8} \sum_{i=1}^3 L_{2,0} ^{(i)} (t, T) + {1\over 4} \sum_{i=1}^4 L_{1,1} ^{(i)} (t, T) - {1\over 8} \sum_{i=1}^3 L_{0,2} ^{(i)} (t, T)\right),
}
where each term is derived in the following form
\eqstar{
L_{2,0} ^{(1)} (t, T) &:= \Psi_1 (\bar{\alpha}, -3, 1, 2, t, T) \partial_x^3\partial_y, 
&& L_{2,0} ^{(2)} (t, T) := \Psi_3 (\bar{\alpha}, -3, 1, 2, t, T) \partial_x\partial_y^3,\\
L_{2,0} ^{(3)} (t, T) &:= 2\Phi_1 (\bar{\alpha}, -3, 1, t, T) \partial_x^2\partial_y^2,
&&L_{0,2} ^{(1)} (t, T) := \Psi_2 (\bar{\alpha}, 1, -3, 2, t, T) \partial_x\partial_y^3,\\ 
L_{0,2} ^{(2)} (t, T) &:= \Psi_3 (\bar{\alpha}, 1, -3, 2, t, T) \partial_x^3\partial_y,
&& L_{0,2} ^{(3)} (t, T) := 2\Phi_2 (\bar{\alpha}, 1, -3, t, T) \partial_x^2\partial_y^2,\\
L_{1,1} ^{(1)} (t, T) &:= \Phi_3 (\bar{\alpha}, -1, -1, t, T) \partial_x^2\partial_y^2, 
&& L_{1,1} ^{(2)} (t, T) := \Phi_1 (\bar{\alpha}, -1, -1, t, T) \partial_x^3\partial_y,\\
L_{1,1} ^{(3)} (t, T) &:= \Phi_2 (\bar{\alpha}, -1, -1, t, T) \partial_x\partial_y^3, 
&& L_{1,1} ^{(4)} (t, T) := \Psi_3 (\bar{\alpha}, -1, -1, 2, t, T) \partial_x^2\partial_y^2.
}
The second part is characterised by
\eqlnostar{eq: L22}{
\Lc_2 ^{(2)} (t, T) := \int_t^T \dd s_1 \Gc_1(t,s_1) \int_{s_1}^T \dd s_2 \Gc_1(t,s_2) = \int_t^T \dd s_1 \Gc_1(t,s_1) \Lc_1 (s_1, T).
}
We will not provide an explicit result for the above integral, as the integrand consists of multiple components, and numerical computation is considered more efficient. By summing \eqref{eq: v0}, \eqref{eq: v1} and \eqref{eq: v2}, we obtain the second-order approximation of $v(t, x, y, T)$, which also represents the second-order approximation of $\Eb \big[ \ee^{-\int_0^{T} (r_s + \lambda_s) \dd s} \big]$.

The second-order approximation of $h(t, x, y, T)$ is derived in a similar way but includes additional terms compared to the second-order approximation of $v(t, x, y, T)$. Given the zeroth term $h_0$ in \eqref{eq: h0}, the first term $h_1(t, x, y, T)$ is given by
\eqstar{
h_1 (t, x, y, T) = \Lc_1 (t, T) h_0 (t, x, y, T) = \Lc_1 (t, T) \bigl( v_0 (t, x, y, T) (y + \alpha_2 \beta_2 \psi(\alpha_2, t,T))\bigr).
}
Let $f: \Rb^2 \to \Rb$ be an infinitely differentiable function. The following property holds
\eqlnostar{eq: partial derivative - y}{
\partial_x^i \partial_y^j\ (f \cdot y) = \big(\partial_x^i \partial_y^j f \big)\cdot y + j \cdot (\partial_x^i \partial_y^{j-1} f), && i,j\in\mathbb{N}_0.
}
Thus, $h_1(t, x, y, T)$ can be rewritten as follows
\eqlnostar{eq: h1}{
h_1 (t, x, y, T) & =\big(\Lc_1 (t, T) v_0 (t, x, y, T) \big) \big(y + \alpha_2 \beta_2 \psi(\alpha_2, t,T)\big)
+ \widetilde{\Lc}_1 (t, T) v_0 (t, x, y, T)\nonumber\\
&= v_1(t, x, y, T) \big(y + \alpha_2 \beta_2 \psi(\alpha_2, t,T)\big) + \widetilde{\Lc}_1 (t, T) v_0 (t, x, y, T),
}
where $\widetilde{\Lc}_1 (t, T)$ denotes the operator for the second term on the right-hand side of \eqref{eq: partial derivative - y}, which is expressed as
\eqstar{
\widetilde{\Lc}_1 (t, T) := \frac{1}{\partial_y} \big(L_{1,0}^{(2)} + L_{1,0}^{(3)} + 2L_{1,0}^{(4)} + L_{1,0}^{(2)} + 3 L_{0,1}^{(1)} + 2 L_{0,1}^{(2)} + 2 L_{0,1}^{(3)} + L_{0,1}^{(4)} + L_{0,1}^{(5)}\big) (t, T).
}
Analogously, the second term in the approximation for $h(t, x, y, T)$ is deduced as follows
\eqlnostar{eq: h2}{
h_2 (t, x, y, T) & =v_2 (t, x, y, T) \big(y + \alpha_2 \beta_2 \psi(\alpha_2, t,T)\big)
+ \widetilde{\Lc}_2 (t, T) v_0 (t, x, y, T).
}
We divide the operator $\widetilde{\Lc}_2 (t, T)$ for the source term into two parts, such that $\widetilde{\Lc}_2 (t, T):= \widetilde{\Lc}_2^{(1)} (t, T) + \widetilde{\Lc}_2^{(2)} (t, T)$. The first component denotes the operator of the extra terms derived from \eqref{eq: L21} according to the principle in \eqref{eq: partial derivative - y}, and the second component denotes that derived from \eqref{eq: L22}. Therefore, $\widetilde{\Lc}_2^{(1)} (t, T)$ is given as
\eqstar{
\widetilde{\Lc}_2^{(1)} (t, T) &:= \frac{1}{\partial_y} \left( L_{2,0} ^{(1)} + 3 L_{2,0} ^{(2)} + 2L_{2,0} ^{(3)} + 3 L_{0,2} ^{(1)} + L_{0,2} ^{(2)} + 2L_{0,2} ^{(3)} \right) (t, T) \\
&+ \frac{1}{\partial_y} \left(2 L_{1,1} ^{(1)} + L_{1,1} ^{(2)} + 3L_{1,1} ^{(3)} + 2L_{1,1} ^{(4)}  \right) (t, T).
}
For $\widetilde{\Lc}_2^{(2)} (t, T)$, we do not present the explicit formula either. By summing \eqref{eq: h0}, \eqref{eq: h1} and \eqref{eq: h2}, we obtain the second-order approximation of $h(t, x, y, T)$.

\subsection{Further calibration results - Bloomberg data}
\label{sec: further calibration - Bloomberg}

In this section, we present the calibration results for two additional entities, Citigroup Inc and Deutsche Bank AG, considering both the correlated and uncorrelated cases.

\begin{table}[H]
\centering
\footnotesize
\captionsetup{font=footnotesize, skip = 5pt}

\begin{threeparttable}
\caption{Calibration results for Citigroup Inc CDS spreads (correlated case)}
\label{tab: citi21}
\begin{tabularx}{.7\textwidth}
{c>{\centering\arraybackslash}X>{\centering\arraybackslash}X>{\centering\arraybackslash}X}
\toprule
& & \multicolumn{2}{c}{PDE approximation}\\
\cmidrule{3-4}
Term (year) & Market (bps) & Model (bps) & Rel. Error \\
\midrule
0.7 & 20.768 & 20.462 & 1.4734\% \\
1.2 & 25.224 & 24.654 & 2.2598\% \\
1.7 & 28.588 & 28.742 & 0.5387\% \\
2.2 & 31.933 & 32.730 & 2.4959\% \\
2.7 & 35.248 & 36.617 & 3.8839\% \\
3.2 & 38.544 & 40.406 & 4.8308\% \\
3.8 & 43.259 & 44.826 & 3.6224\% \\
4.3 & 47.974 & 48.400 & 0.8880\% \\
4.8 & 52.328 & 51.876 & 0.8638\% \\
5.3 & 56.658 & 55.255 & 2.4763\% \\
5.8 & 60.104 & 58.540 & 2.6022\% \\
6.3 & 63.531 & 61.729 & 2.8364\% \\
6.8 & 66.977 & 64.825 & 3.2130\% \\
7.3 & 70.404 & 67.829 & 3.6575\% \\
7.8 & 72.211 & 70.742 & 2.0343\% \\
8.3 & 74.019 & 73.566 & 0.6120\% \\
8.8 & 75.826 & 76.302 & 0.6278\% \\
9.3 & 77.623 & 78.953 & 1.7134\% \\
9.8 & 79.431 & 81.519 & 2.6287\% \\
10.3 & 81.228 & 84.003 & 3.4163\% \\
\bottomrule
\end{tabularx}
\begin{tablenotes}
\item $\alpha_2 = 0.04372$, $\beta_2 = 0.06900$, $\sigma_2 = 0.06852$, $\lambda_0 = 0.00239$, $\rho = -0.69724$.
\end{tablenotes}
\end{threeparttable}
\end{table}
\begin{table}[H]
\centering
\footnotesize
\captionsetup{font=footnotesize, skip = 5pt}
\begin{threeparttable}
\caption{Calibration results for Deutsche Bank AG CDS spreads (correlated case)}
\label{tab: deut21}
\footnotesize
\begin{tabularx}{.7\textwidth}
{c>{\centering\arraybackslash}X>{\centering\arraybackslash}X>{\centering\arraybackslash}X}
\toprule
& & \multicolumn{2}{c}{PDE approximation}\\
\cmidrule{3-4}
Term (year) & Market (bps) & Model (bps) & Rel. Error \\
\midrule
0.7 & 25.341 & 25.282 & 0.2328\% \\
1.2 & 30.956 & 30.983 & 0.0872\% \\
1.7 & 36.652 & 36.633 & 0.0518\% \\
2.2 & 42.317 & 42.233 & 0.1985\% \\
2.7 & 47.931 & 47.781 & 0.3130\% \\
3.2 & 53.514 & 53.275 & 0.4466\% \\
3.8 & 58.108 & 59.791 & 2.8963\% \\
4.3 & 62.701 & 65.156 & 3.9154\% \\
4.8 & 68.727 & 70.459 & 2.5201\% \\
5.3 & 74.721 & 75.697 & 1.3062\% \\
5.8 & 81.100 & 80.869 & 0.2848\% \\
6.3 & 87.444 & 85.972 & 1.6834\% \\
6.8 & 93.824 & 91.005 & 3.0046\% \\
7.3 & 100.168 & 95.966 & 4.1950\% \\
7.8 & 103.632 & 100.853 & 2.6816\% \\
8.3 & 107.095 & 105.664 & 1.3362\% \\
8.8 & 110.558 & 110.398 & 0.1447\% \\
9.3 & 114.003 & 115.053 & 0.9210\% \\
9.8 & 117.466 & 119.629 & 1.8414\% \\
10.3 & 120.911 & 124.123 & 2.6565\% \\
\bottomrule
\end{tabularx}
\begin{tablenotes}
\item $\alpha_2 = 0.00777$, $\beta_2 = 0.49897$, $\sigma_2 = 0.00014$, $\lambda_0 = 0.00286$, $\rho = 0.34303$.
\end{tablenotes}
\end{threeparttable}
\end{table}

\begin{table}[H]
\centering
\footnotesize
\captionsetup{font=footnotesize, skip = 5pt}

\begin{threeparttable}
\caption{Calibration results for Citigroup Inc CDS spreads (uncorrelated case)}
\label{tab: citi22}
\begin{tabularx}{.7\textwidth}
{c>{\centering\arraybackslash}X>{\centering\arraybackslash}X>{\centering\arraybackslash}X}
\toprule
& & \multicolumn{2}{c}{PDE approximation}\\
\cmidrule{3-4}
Term (year) & Market (bps) & Model (bps) & Rel. Error \\
\midrule
0.7 & 20.768 & 20.435 & 1.6034\% \\
1.2 & 25.224 & 24.640 & 2.3153\% \\
1.7 & 28.588 & 28.738 & 0.5247\% \\
2.2 & 31.933 & 32.734 & 2.5084\% \\
2.7 & 35.248 & 36.628 & 3.9151\% \\
3.2 & 38.544 & 40.421 & 4.8698\% \\
3.8 & 43.259 & 44.844 & 3.6640\% \\
4.3 & 47.974 & 48.419 & 0.9276\% \\
4.8 & 52.328 & 51.895 & 0.8275\% \\
5.3 & 56.658 & 55.274 & 2.4427\% \\
5.8 & 60.104 & 58.555 & 2.5772\% \\
6.3 & 63.531 & 61.742 & 2.8160\% \\
6.8 & 66.977 & 64.833 & 3.2011\% \\
7.3 & 70.404 & 67.832 & 3.6532\% \\
7.8 & 72.211 & 70.739 & 2.0385\% \\
8.3 & 74.019 & 73.557 & 0.6242\% \\
8.8 & 75.826 & 76.285 & 0.6053\% \\
9.3 & 77.623 & 78.928 & 1.6812\% \\
9.8 & 79.431 & 81.486 & 2.5872\% \\
10.3 & 81.228 & 83.961 & 3.3646\% \\
\bottomrule
\end{tabularx}
\begin{tablenotes}
\item $\alpha_2 = 0.04539$, $\beta_2 = 0.06678$, $\sigma_2 = 0.06657$, $\lambda_0 = 0.00238$.
\end{tablenotes}
\end{threeparttable}
\end{table}
\begin{table}[H]
\centering
\footnotesize
\captionsetup{font=footnotesize, skip = 5pt}
\begin{threeparttable}
\caption{Calibration results for Deutsche Bank AG CDS spreads (uncorrelated case)}
\label{tab: deut22}
\begin{tabularx}{.7\textwidth}
{c>{\centering\arraybackslash}X>{\centering\arraybackslash}X>{\centering\arraybackslash}X}
\toprule
& & \multicolumn{2}{c}{PDE approximation}\\
\cmidrule{3-4}
Term (year) & Market (bps) & Model (bps) & Rel. Error \\
\midrule
0.7 & 25.341 & 25.121 & 0.8682\% \\
1.2 & 30.956 & 30.898 & 0.1874\% \\
1.7 & 36.652 & 36.613 & 0.1064\% \\
2.2 & 42.317 & 42.269 & 0.1134\% \\
2.7 & 47.931 & 47.863 & 0.1419\% \\
3.2 & 53.514 & 53.392 & 0.2280\% \\
3.8 & 58.108 & 59.939 & 3.1510\% \\
4.3 & 62.701 & 65.319 & 4.1754\% \\
4.8 & 68.727 & 70.627 & 2.7646\% \\
5.3 & 74.721 & 75.862 & 1.5270\% \\
5.8 & 81.100 & 81.022 & 0.0962\% \\
6.3 & 87.444 & 86.104 & 1.5324\% \\
6.8 & 93.824 & 91.108 & 2.8948\% \\
7.3 & 100.168 & 96.031 & 4.1301\% \\
7.8 & 103.632 & 100.872 & 2.6633\% \\
8.3 & 107.095 & 105.629 & 1.3689\% \\
8.8 & 110.558 & 110.303 & 0.2307\% \\
9.3 & 114.003 & 114.890 & 0.7781\% \\
9.8 & 117.466 & 119.390 & 1.6379\% \\
10.3 & 120.911 & 123.803 & 2.3918\% \\
\bottomrule
\end{tabularx}
\begin{tablenotes}
\item $\alpha_2 = 0.01235$, $\beta_2 = 0.31977$, $\sigma_2 = 0.01622$, $\lambda_0 = 0.00282$.
\end{tablenotes}
\end{threeparttable}
\end{table}

\subsection{Further calibration results - negative interest rates}
\label{sec: further calibration - JDCEV}

In this section, we present the calibration results based on the market data reported in \citet{di_francesco_cds_2019}, for four distinct entities: Caixa Bank SA, Commerzbank AG, Deutsche Bank AG, and Mediobanca SpA.

\begin{table}[H]
\centering
\footnotesize
\captionsetup{font=footnotesize, skip = 5pt}
\begin{threeparttable}
\caption{Calibration results for Caixa Bank SA CDS spreads (correlated case)}
\label{tab: caixa1}
\begin{tabularx}{.7\textwidth}
{c>{\centering\arraybackslash}X>{\centering\arraybackslash}X>{\centering\arraybackslash}X>{\centering\arraybackslash}X>{\centering\arraybackslash}X}
\toprule
Term & Market & JDCEV & Rel. Error & SSRD & Rel. Error\\
(year) & (bps) & (bps) &  & (bps) &  \\
\midrule
1.0 & 76.878 & 76.655 & 0.29046\% & 75.872 & 1.30895\% \\
1.5 & 83.557 & 85.162 & 1.92060\% & 83.511 & 0.05553\% \\
2.0 & 90.452 & 90.115 & 0.37224\% & 91.107 & 0.72448\% \\
2.5 & 97.552 & 96.184 & 1.40223\% & 98.652 & 1.12802\% \\
3.0 & 104.847 & 103.465 & 1.31811\% & 106.136 & 1.22941\% \\
3.5 & 112.325 & 111.251 & 0.95615\% & 113.553 & 1.09326\% \\
4.0 & 119.976 & 120.190 & 0.17837\% & 120.895 & 0.76599\% \\
4.5 & 127.787 & 130.524 & 2.14185\% & 128.155 & 0.28798\% \\
5.0 & 135.743 & 139.515 & 2.77878\% & 135.327 & 0.30646\% \\
5.5 & 143.832 & 144.723 & 0.61947\% & 142.406 & 0.99143\% \\
6.0 & 152.039 & 147.885 & 2.73219\% & 149.385 & 1.74560\% \\
\bottomrule
\end{tabularx}
\begin{tablenotes}
\item $\alpha_2 = 0.00561$, $\beta_2 = 0.92493$, $\sigma_2 = 0.02352$, $\lambda_0 = 0.01011$, $\rho = -0.02910$.
\end{tablenotes}
\end{threeparttable}
\end{table}
\begin{table}[H]
\centering
\footnotesize
\captionsetup{font=footnotesize, skip = 5pt}
\begin{threeparttable}
\caption{Calibration results for Commerzbank AG CDS spreads (correlated case)}
\label{tab: commer1}
\begin{tabularx}{.7\textwidth}
{c>{\centering\arraybackslash}X>{\centering\arraybackslash}X>{\centering\arraybackslash}X>{\centering\arraybackslash}X>{\centering\arraybackslash}X}
\toprule
Term & Market & JDCEV & Rel. Error & SSRD & Rel. Error\\
(year) & (bps) & (bps) &  & (bps) &  \\
\midrule
1.0 & 44.690 & 44.782 & 0.20564\% & 44.794 & 0.23271\% \\
1.5 & 53.933 & 54.045 & 0.20711\% & 53.932 & 0.00148\% \\
2.0 & 63.175 & 63.075 & 0.15877\% & 62.925 & 0.39573\% \\
2.5 & 71.838 & 71.884 & 0.06515\% & 71.764 & 0.10245\% \\
3.0 & 80.285 & 80.491 & 0.25671\% & 80.442 & 0.19555\% \\
3.5 & 88.977 & 88.916 & 0.06844\% & 88.953 & 0.02697\% \\
4.0 & 97.810 & 97.184 & 0.64053\% & 97.291 & 0.53062\% \\
4.5 & 106.475 & 105.319 & 1.08570\% & 105.451 & 0.96173\% \\
5.0 & 114.405 & 113.349 & 0.92304\% & 113.430 & 0.85224\% \\
5.5 & 121.117 & 121.298 & 0.14944\% & 121.225 & 0.08917\% \\
6.0 & 126.730 & 129.192 & 1.94271\% & 128.833 & 1.65943\% \\
\bottomrule
\end{tabularx}
\begin{tablenotes}
\item $\alpha_2 = 0.03966$, $\beta_2 = 0.16350$, $\sigma_2 = 0.01600$, $\lambda_0 = 0.00436$, $\rho = 0.04662$.
\end{tablenotes}
\end{threeparttable}
\end{table}
\begin{table}[H]
\centering
\footnotesize
\captionsetup{font=footnotesize, skip = 5pt}
\begin{threeparttable}
\caption{Calibration results for Deutsche Bank AG CDS spreads (correlated case)}
\label{tab: deut1}
\begin{tabularx}{.7\textwidth}
{c>{\centering\arraybackslash}X>{\centering\arraybackslash}X>{\centering\arraybackslash}X>{\centering\arraybackslash}X>{\centering\arraybackslash}X}
\toprule
Term & Market & JDCEV & Rel. Error & SSRD & Rel. Error\\
(year) & (bps) & (bps) &  & (bps) &  \\
\midrule
1.0 & 67.020 & 66.420 & 0.89526\% & 65.515 & 2.24560\% \\
1.5 & 77.786 & 79.185 & 1.79774\% & 80.347 & 3.29184\% \\
2.0 & 92.015 & 92.529 & 0.55850\% & 94.052 & 2.21377\% \\
2.5 & 107.233 & 105.796 & 1.34007\% & 106.705 & 0.49239\% \\
3.0 & 120.505 & 118.488 & 1.67379\% & 118.378 & 1.76507\% \\
3.5 & 129.985 & 130.235 & 0.19233\% & 129.141 & 0.64931\% \\
4.0 & 138.645 & 140.774 & 1.53558\% & 139.060 & 0.29933\% \\
4.5 & 149.244 & 149.928 & 0.45831\% & 148.200 & 0.69953\% \\
5.0 & 158.860 & 157.592 & 0.79819\% & 156.617 & 1.41194\% \\
5.5 & 164.376 & 163.718 & 0.40030\% & 164.369 & 0.00426\% \\
6.0 & 167.590 & 168.305 & 0.42664\% & 171.508 & 2.33785\% \\
\bottomrule
\end{tabularx}
\begin{tablenotes}
\item $\alpha_2 = 0.22724$, $\beta_2 = 0.05817$, $\sigma_2 = 0.06869$, $\lambda_0 = 0.00537$, $\rho = -0.05432$.
\end{tablenotes}
\end{threeparttable}
\end{table}
\begin{table}[H]
\centering
\footnotesize
\captionsetup{font=footnotesize, skip = 5pt}
\begin{threeparttable}
\caption{Calibration results for Mediobanca SpA CDS spreads (correlated case)}
\label{tab: medio1}
\begin{tabularx}{.7\textwidth}
{c>{\centering\arraybackslash}X>{\centering\arraybackslash}X>{\centering\arraybackslash}X>{\centering\arraybackslash}X>{\centering\arraybackslash}X}
\toprule
Term & Market & JDCEV & Rel. Error & SSRD & Rel. Error\\
(year) & (bps) & (bps) &  & (bps) &  \\
\midrule
1.00 & 87.545 & 87.456 & 0.10200\% & 87.063 & 0.55057\% \\
1.50 & 96.831 & 96.862 & 0.03232\% & 97.283 & 0.46679\% \\
2.00 & 106.715 & 106.688 & 0.02530\% & 107.292 & 0.54069\% \\
2.50 & 116.657 & 116.678 & 0.01800\% & 117.079 & 0.36174\% \\
3.00 & 126.405 & 126.621 & 0.17088\% & 126.633 & 0.18037\% \\
3.50 & 135.880 & 136.343 & 0.34074\% & 135.949 & 0.05078\% \\
4.00 & 145.420 & 145.703 & 0.19461\% & 145.019 & 0.27575\% \\
4.50 & 155.159 & 154.586 & 0.36930\% & 153.837 & 0.85203\% \\
5.00 & 164.010 & 162.903 & 0.67496\% & 162.403 & 0.97982\% \\
5.50 & 170.956 & 170.583 & 0.21818\% & 170.713 & 0.14214\% \\
6.00 & 176.485 & 177.572 & 0.61592\% & 178.765 & 1.29189\% \\
\bottomrule
\end{tabularx}
\begin{tablenotes}
\item $\alpha_2 = 0.04117$, $\beta_2 = 0.18416$, $\sigma_2 = 0.07196$, $\lambda_0 = 0.01103$, $\rho = 0.05469$.
\end{tablenotes}
\end{threeparttable}
\end{table}

\end{appendix}

\end{document}